\documentclass[11pt]{article}

\usepackage[margin=1in]{geometry}
\usepackage[T1]{fontenc}
\usepackage{lmodern}
\usepackage{amsmath,amssymb,amsthm,mathtools}
\usepackage{algorithm}
\usepackage{algpseudocode}
\usepackage{thmtools}
\makeatletter
\@ifundefined{newcounteralias}{}{%
  \renewcommand\thmt@autorefsetup{%
    \@xa\def\csname\thmt@envname autorefname\@xa\endcsname
      \@xa{\thmt@thmname}%
  }%
}
\makeatother
\usepackage{booktabs}
\usepackage{enumitem}
\usepackage{xcolor}
\usepackage{tikz}
\usetikzlibrary{arrows.meta,calc,positioning}
\usepackage{hyperref}
\usepackage{microtype}
\usepackage{verbatim}
\usepackage{dsfont}
\usepackage[capitalise,nameinlink]{cleveref}

\hypersetup{
  colorlinks=true,
  linkcolor=blue!50!black,
  citecolor=blue!50!black,
  urlcolor=blue!50!black
}

\newtheorem{theorem}{Theorem}[section]
\newtheorem{lemma}[theorem]{Lemma}

\newtheorem{corollary}[theorem]{Corollary}

\newcommand{\OPT}{\operatorname{OPT}}
\newcommand{\cost}{\operatorname{cost}}
\newcommand{\rad}{\operatorname{rad}}
\newcommand{\val}{\operatorname{val}}

\newcommand{\U}{\mathcal U}

\newcommand{\supp}{\operatorname{supp}}
\newcommand{\ED}{\operatorname{ED}}

\newcommand{\parity}{\mathrm{Parity}}

\newcommand{\idc}[1]{\mathds{1}_{#1}} 
\newcommand{\LCS}{\mathsf{LCS}}
\newcommand{\Agr}{\mathsf{Agr}}
\newcommand{\Val}{\operatorname{Val}}

\DeclareMathOperator*{\argmin}{arg\,min}

\definecolor{greenlantern}{rgb}{0.0, 0.5, 0.0}

\definecolor{red}{rgb}{1.0, 0.0, 0.0}

\title{Ulam Rank Aggregation Is Hard to Approximate for Four Rankings\footnote{This paper merges two recent works~\cite{Habib26, azgor2026hardness}, subsuming and strengthening the results of both.}}

\author{
Sk Ruhul Azgor%
  \thanks{Pennsylvania State University.
    Email: \texttt{sfa6135@psu.edu}
  }
\and
Diptarka Chakraborty%
\thanks{National University of Singapore.
  Work partially supported by an MoE AcRF Tier 1 grant (T1 251RES2303) and a Google South \& South-East Asia Research Award.
    Email: \texttt{diptarka@nus.edu.sg}
  }
\and
Le Van Cuong%
\thanks{National University of Singapore.
    Email: \texttt{e1583414@u.nus.edu}
  }
\and
Debarati Das%
\thanks{Pennsylvania State University.
Work supported in part by NSF grant 2337832.
        Email: \texttt{debaratix710@gmail.com}
}
\and
Mursalin Habib%
\thanks{Rutgers University. Supported by the National Science Foundation under Grants CCF-2313372 and CCF-2443697.
        Email: \texttt{mursalin.habib@rutgers.edu}
}
\and
Tien Long Nguyen%
  \thanks{Pennsylvania State University.
    Email: \texttt{tfn5179@psu.edu}
  }
}

\date{}

\begin{document}
\maketitle

\begin{abstract}
    We study the approximability of rank aggregation under the Ulam metric. In the \emph{Ulam median} problem, the goal is to find a ranking (permutation) minimizing the sum of its Ulam distances to the input rankings, while in the \emph{Ulam center} problem, the objective is to minimize the maximum such distance. We prove that, for every $0<\varepsilon< 1/34$, it is $\mathrm{NP}$-hard to approximate either Ulam median or Ulam center within a factor of $35/34-\varepsilon$, even when the input consists of only four rankings. We further show that unless P = NP, neither problem admits a polynomial-time additive approximation scheme. Prior to our work, only the exact versions of both problems were known to be $\mathrm{NP}$-hard, and that too only when the number of input rankings is unbounded [Fischer et al., ESA'25 and Bachmaier et al., J. of Discrete Algorithms'15]. Furthermore, our inapproximability results are optimal in terms of the number of input rankings since for three inputs it is already known to be polynomial-time solvable [Chakraborty, Das, Krauthgamer, SODA'21].
    
    En route, we introduce a new general framework for reducing Boolean constraint satisfaction problems (CSP) to the Ulam median with only four inputs. As a specific instantiation of the reduction framework, we obtain our hardness-of-approximation results. The corresponding hardness for the Ulam center follows from a reduction from the Ulam median.
\end{abstract}

\section{Introduction}\label{sec:introduction}

Aggregating potentially inconsistent information from multiple sources is a fundamental challenge across social choice, information retrieval, and data analysis. When sources such as voters, search engines, or learned agents provide complete orderings over a set of alternatives, the problem is known as \emph{rank aggregation}. The goal is to compute a consensus ranking that optimizes a specific objective function. Because different individuals or retrieval systems rarely agree on every pair of alternatives, input rankings inherently conflict. This setting naturally arises in diverse applications, including voting, web search, database systems, similarity search, and classification~\cite{brandt2016handbook, DKNS01, Harman92a, FKS03, ACN08}.

Two of the most well-studied optimization objectives in this domain are the \emph{median} and \emph{center} variants. The median variant seeks a consensus ranking (permutation) that minimizes the sum of distances to the input rankings, while the center variant minimizes the maximum distance to any input. In both cases, the optimal consensus may be any valid permutation, not strictly one of the original inputs. Ultimately, the computational complexity of solving either variant depends heavily on the underlying distance metric used to compare permutations.

In this paper, we study rank aggregation under the \emph{Ulam distance}. For two rankings of $d$ items, their Ulam distance is the minimum number of move operations required to transform one into the other. Equivalently, it is $d$ minus the length of their \emph{longest common subsequence (LCS)}, or one-half of their insertion-deletion (edit) distance, thereby retaining the alignment structure of edit distance without repeated symbols. The metric is also natural for rankings: a move models a displaced item while preserving the relative order of a large unaffected subsequence. Accordingly, Ulam distance has been studied extensively in several algorithmic settings~\cite{CMS01,CK06,AN10,NSS17,BS19}. Rank aggregation under the Ulam metric, in particular, has attracted considerable recent attention~\cite{chakraborty2021approximating,CGJ21,CDK23,jaiswal2025robust,CDN26,bai26}.

For an arbitrary number of input rankings, the Ulam median problem is $\mathrm{NP}$-hard~\cite{fischer2025} and the best-known polynomial-time approximation factor for Ulam median is \(1.968\)~\cite{CDN26}, building on a sequence of earlier advances~\cite{chakraborty2021approximating,CDK23,jaiswal2025robust}. The $\mathrm{NP}$-hardness result in~\cite{fischer2025}, however, applies only when the number of input permutations is unbounded, that is, when the number of inputs is part of the instance. It therefore yields no hardness for any fixed number of inputs. At the other end of the spectrum, Ulam median can be solved exactly in polynomial time for three input permutations~\cite{chakraborty2021approximating}. Thus, prior to this work, the complexity of the problem remained open for every fixed number \(k\geq 4\) of inputs.  

This gap is particularly intriguing given the close connection between the Ulam and edit metrics. For the more general \emph{median-string problem} under insertion-deletion edit distance, where both the input strings and the median may be arbitrary strings, an exact median can be computed in polynomial time for every fixed number of inputs using a standard multidimensional dynamic program~\cite{Sankoff75,kruskal1983,NR03}. Since the Ulam distance is precisely one-half of the insertion-deletion distance on permutations over a common alphabet, one might expect the additional permutation structure to simplify the median problem. This leads to the following question:

\vspace{-2mm}

\begin{center}
\emph{Can the Ulam median be computed exactly in polynomial time for every fixed number of input permutations?}
\end{center}

\vspace{-2mm}


As mentioned earlier, when the number of input rankings is unrestricted, the best-known polynomial-time approximation factor for Ulam median is \(1.968\). The situation appears more favorable when the number of inputs is fixed. Using the fixed-dimensional dynamic program for edit-distance median as a relaxation and subsequently transforming its output into a permutation yields a polynomial-time approximation factor strictly better than \(3/2\) for every fixed number of input permutations. In particular, a sharper analysis of the approach of~\cite{chakraborty2021approximating} gives a \(5/4\)-approximation for four inputs. For completeness, we present the algorithm and its analysis in~\Cref {sec:approx-constant}.

Further evidence that arbitrarily accurate approximation might be possible comes from the Kendall-tau metric, another standard distance measure on rankings. The median problem under Kendall-tau, also known as \emph{Kemeny rank aggregation}, is \(\mathrm{NP}\)-hard already for three input rankings~\cite{peters2026kemeny,madarasi2026complexity}, yet admits a polynomial-time approximation scheme (PTAS) even when the number of inputs is unrestricted~\cite{kenyon2007rank}. Thus, although Ulam median is easier from the viewpoint of exact computation for three inputs, it is natural to ask whether it enjoys a comparable approximation guarantee. This leads to our second question:

\vspace{-2mm}

\begin{center}
\emph{Does Ulam median admit a PTAS, at least when the number of input permutations is fixed?}
\end{center}

\vspace{-2mm}

We answer both questions negatively. We prove an explicit constant-factor hardness of approximation for Ulam median with exactly four input permutations, thereby ruling out a PTAS unless \(\mathrm{P}=\mathrm{NP}\). Together with the exact algorithm for three inputs, this establishes a sharp transition from exact polynomial-time solvability with three inputs to constant-factor inapproximability with four. It also reverses the comparison with Kendall-tau median from the viewpoint of approximation: whereas Kendall-tau median admits a PTAS even for an unrestricted number of inputs, Ulam median admits no PTAS already for four. Thus, despite being easier for exact computation on three inputs, Ulam median is strictly harder to approximate in this precise sense. Finally, our reduction from Ulam median to Ulam center transfers the same inapproximability to the center objective.

\subsection{Our results}

Our main result gives an explicit multiplicative inapproximability bound for the Ulam median with only four input permutations. It rules out a PTAS unless \(\mathrm{P}=\mathrm{NP}\), in contrast to the PTAS for Kendall-tau median~\cite{kenyon2007rank}.

\begin{restatable}{theorem}{FourInputHardnessTheorem}
\label{thm:four-input-hardness}
    For every fixed \(0<\varepsilon<1/34\), it is $\mathrm{NP}$-hard to approximate the Ulam median within a (multiplicative) factor of \(35/34-\varepsilon\), even when the input consists of exactly four permutations.
\end{restatable}

We would like to emphasize that the previously known $\mathrm{NP}$-hardness reduction for the Ulam median requires an unbounded number of input permutations~\cite{fischer2025}. It starts from
\textsc{Max-Cut}, viewed as a Boolean CSP in which vertices are
variables and edges are constraints. It introduces a pair of
permutations for every edge, so the number of input permutations
depends on the number of constraints. It also uses many auxiliary
permutations to ensure that the median encodes a valid cut. These
auxiliary permutations dominate the objective and destroy any gap
in the source \textsc{Max-Cut} instance. Consequently, the
reduction establishes $\mathrm{NP}$-hardness but does not yield any constant factor hardness of approximation.

In contrast, our reduction uses only two permutations to encode
all the constraints and two more to enforce consistency. This
keeps the number of input permutations at four while preserving the
approximation gap. More generally, we develop a framework that transforms local permutation gadgets for Boolean constraints into gap-preserving reductions to the Ulam median with only four input permutations (see~\Cref{sec:boolean-csp-framework}). The framework allows the same consistency and anchor constructions to be reused across different CSPs. Combining polynomial-time constructible, constant-size labeled gadgets with prescribed LCS scores for satisfying and unsatisfying assignments and globally balanced labels for each variable, we establish an exact relation between the CSP optimum and the Ulam-median optimum. Instantiating this framework with \textsc{Max-E3-Lin-2} yields~\Cref{thm:four-input-hardness}. 

The same reduction produces a gap linear in the total input length \(kd\) (for $k$ input permutations over $d$ symbols). This also rules out a polynomial-time additive approximation scheme unless \(\mathrm{P}=\mathrm{NP}\), even for four input permutations.

\begin{restatable}{theorem}{MedianAdditiveHardness}
\label{thm:median-additive-hardness}
    For every fixed $0 < \varepsilon \leq 1/200$, it is $\mathrm{NP}$-hard to approximate the Ulam median within an additive $\varepsilon kd$-factor with $k$ input permutations over an alphabet of size $d$, even when $k = 4$.
\end{restatable}

Together with the exact polynomial-time algorithm for three input permutations~\cite{chakraborty2021approximating}, our result establishes a sharp transition from exact tractability with three inputs to constant-factor inapproximability with four. By comparison, the general median-string problem under edit distance can be solved exactly in polynomial time for every fixed number of input strings~\cite{Sankoff75,kruskal1983}. This contrast highlights the computational challenges imposed by requiring the median itself to be a permutation, even when the number of inputs is fixed.

The same multiplicative inapproximability bound holds for the Ulam center, ruling out a PTAS unless \(\mathrm{P}=\mathrm{NP}\), again with only four input permutations.

\begin{restatable}{theorem}{CenterHardnessTheorem}
\label{thm:center-hardness}
    For every fixed \(0<\varepsilon<1/34\), it is $\mathrm{NP}$-hard to approximate the Ulam center within a (multiplicative) factor of \(35/34-\varepsilon\), even when the input consists of exactly four permutations.
\end{restatable}

Previously, the Ulam center was known to be $\mathrm{NP}$-hard only when the number of inputs is unrestricted~\cite{BACHMAIER20152}. On the algorithmic side, it admits an exact polynomial-time algorithm for three input permutations and an approximation factor strictly better than $3/2$ for every fixed number of inputs~\cite{CGJ21}. Our result therefore establishes the same sharp transition from exact tractability with three inputs to constant-factor inapproximability with four. For general strings under edit distance, however, a dynamic program analogous to the one for median strings computes an exact center in polynomial time for any fixed number of inputs. Thus, the hardness caused by the permutation constraint extends beyond the median objective and holds for the center objective as well.

We obtain the hardness for the center objective through an approximation-preserving reduction from Ulam median to Ulam center that also preserves the number of input permutations. Consequently, any improvement to the multiplicative hardness bound for Ulam median immediately transfers to Ulam center under the same restriction on the number of inputs. The same reduction also rules out a polynomial-time additive approximation scheme for Ulam center unless \(\mathrm{P}=\mathrm{NP}\), even with four inputs.

\begin{restatable}{theorem}{CenterAdditiveHardness}
\label{thm:center-additive-hardness}
For every fixed $0 < \varepsilon \leq 1/200$, it is $\mathrm{NP}$-hard to approximate the Ulam center within an additive $\varepsilon d$-factor with $k$ input permutations over an alphabet of size $d$, even when $k = 4$.
\end{restatable}

\subsection{Technical overview}

We begin with an overview of the reduction establishing $\mathrm{NP}$-hardness, then explain how its underlying ideas lead to a general framework and our hardness-of-approximation results.

\paragraph{The basic 3-SAT reduction.} To illustrate our main techniques, we start with a simplified description
of how a \textsc{3-SAT} formula \(\varphi\) can be transformed into
four permutations whose optimal Ulam median encodes its
satisfiability. The key is to view satisfiability as a \textit{partition
problem}: place every literal occurrence on a true side or a false
side, so that every clause contains a true occurrence and, for each
variable, all positive occurrences lie on one side and all negative
occurrences on the other. The latter condition ensures that the
partition comes from a single assignment. Thus, our goal is to construct four permutations whose Ulam median encodes this partition problem.

To do so, we start from a \emph{balanced}\footnote{Balanced means that every variable has equally
many positive and negative occurrences.} \(3\)-CNF formula
\(\varphi\) and introduce a fresh symbol for each \textit{literal occurrence} (so even
two appearances of the same literal in different clauses receive different symbols).
Next, we construct two permutations \(\beta_1,\beta_2\) over this alphabet that give
a \emph{clause score}. For a set \(S\) of proposed true occurrences (which are possibly inconsistent),
they satisfy
\[
\LCS(\beta_1|_S,\beta_2|_S)
=
\text{number of clauses containing an occurrence from \(S\)},
\]
where by \(\beta_i|_S\), we mean the string obtained by deleting all symbols outside \(S\) in \(\beta_i\). {To construct this pair, let \(a_j,b_j,c_j\) denote the three occurrence symbols of clause \(j\). In a fixed clause order, \(\beta_1\) concatenates the blocks \(a_jb_jc_j\), while \(\beta_2\) concatenates their reversals \(c_jb_ja_j\). A common subsequence can therefore contain at most one symbol from each clause, and the shared block order lets it contain one from every clause with a symbol in \(S\).}

Then, we construct two more permutations \(\alpha_1,\alpha_2\) that give
a \emph{consistency score}, measuring how consistent different occurrences of the same literal are across clauses. For a set \(S'\)
of proposed false occurrences, they satisfy
\[
\LCS(\alpha_1|_{S'},\alpha_2|_{S'})
=
\begin{gathered}
\text{maximum number of occurrences in \(S'\) that can}\\
\text{simultaneously be false under a single assignment}.
\end{gathered}
\]
{To construct this pair, let \(B_x^0\) and \(B_x^1\) list all negative and all positive occurrences of \(x\), respectively, in fixed internal orders. Both permutations list the variables in the same order, using \(B_x^0B_x^1\) in \(\alpha_1\) and \(B_x^1B_x^0\) in \(\alpha_2\), without changing either block internally. A common subsequence can use symbols from at most one block for each variable, but can retain all selected symbols from that block. After restricting to \(S'\), choosing the block with more remaining symbols corresponds to setting \(x\) so that as many of its selected occurrences as possible are false, yielding the stated consistency score.}
Since the formula is balanced, every assignment makes exactly half
of all occurrences false. This is therefore an upper bound on the
consistency score, attained by the false side of any consistent
partition. \Cref{fig:overview-local-gadgets} illustrates the two mechanisms: reversing the occurrence order within each clause block and swapping the blocks of positive and negative occurrences for each variable.

\begin{figure}[!htbp]
\centering
\begin{tikzpicture}[
  font=\small,
  token/.style={draw=black!40,fill=black!3,minimum width=.78cm,
    minimum height=.58cm,inner sep=2pt},
  chosen/.style={token,draw=blue!65!black,fill=blue!10},
  block/.style={token,minimum width=1.65cm},
  picked/.style={block,draw=orange!70!black,fill=orange!15},
  rowlabel/.style={anchor=east,font=\scriptsize}
]
  \node[font=\small\bfseries] at (2.8,.65) {(a) Clause witnesses};
  \node[rowlabel] at (.65,0) {$\beta_1$};
  \node[rowlabel] at (.65,-1.3) {$\beta_2$};
  \node at (1.05,0) {$\cdots$};
  \node at (1.05,-1.3) {$\cdots$};
  \node[token] (a1) at (1.8,0) {$a_j$};
  \node[chosen] (b1) at (2.8,0) {$b_j$};
  \node[token] (c1) at (3.8,0) {$c_j$};
  \node[token] (c2) at (1.8,-1.3) {$c_j$};
  \node[chosen] (b2) at (2.8,-1.3) {$b_j$};
  \node[token] (a2) at (3.8,-1.3) {$a_j$};
  \node at (4.6,0) {$\cdots$};
  \node at (4.6,-1.3) {$\cdots$};
  \draw[black!25,dashed] (a1.south)--(a2.north);
  \draw[black!25,dashed] (c1.south)--(c2.north);
  \draw[blue!65!black,line width=1pt] (b1.south)--(b2.north);
  \node[align=center,text width=5.6cm] at (2.8,-2.15)
    {At most one LCS match\\from each clause};

  \draw[black!20] (6.1,.95)--(6.1,-2.65);
  \begin{scope}[xshift=7.1cm]
    \node[font=\small\bfseries] at (2.8,.65) {(b) Variable consistency};
    \node[rowlabel] at (.3,0) {$\alpha_1$};
    \node[rowlabel] at (.3,-1.3) {$\alpha_2$};
    \node at (.75,0) {$\cdots$};
    \node at (.75,-1.3) {$\cdots$};
    \node[picked] (n1) at (1.9,0) {$B_x^0$};
    \node[block] (p1) at (3.8,0) {$B_x^1$};
    \node[block] (p2) at (1.9,-1.3) {$B_x^1$};
    \node[picked] (n2) at (3.8,-1.3) {$B_x^0$};
    \node at (4.95,0) {$\cdots$};
    \node at (4.95,-1.3) {$\cdots$};
    \draw[black!25,dashed] (p1.south)--(p2.north);
    \draw[orange!70!black,line width=1pt] (n1.south)--(n2.north);
    \node[align=center,text width=5.6cm] at (2.8,-2.15)
      {At most one label class\\from each variable};
  \end{scope}
\end{tikzpicture}
\caption{The two local mechanisms in the \textsc{3-SAT} reduction.
(a) The three occurrence symbols of each clause appear in reverse order,
while the clause blocks retain their order. Restricting to $S$ gives one
LCS match for each clause meeting $S$.
(b) For each variable $x$, the negative-occurrence block $B_x^0$ and the
positive-occurrence block $B_x^1$ swap places, retaining their internal
orders. Restricting to $S'$ contributes the larger of the two label-class
sizes. Solid connections illustrate one admissible choice in each panel.}
\label{fig:overview-local-gadgets}
\end{figure}
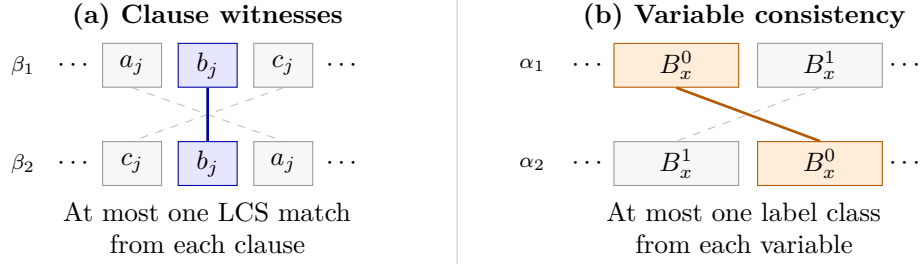

Finally, we take a long anchor block \(Z\) of fresh auxiliary
symbols and form the four input permutations
\[
\beta_1Z,\qquad \beta_2Z,\qquad Z\alpha_1,\qquad Z\alpha_2.
\]
Our anchor-block lemma (\Cref{lem:anchor-block}) shows that, when \(Z\) is sufficiently long,
some optimal median has the form \(\tau_S Z\tau_{S'}\), where
\(\tau_S\) and \(\tau_{S'}\) order the symbols in a partition
\(S\sqcup S'\) of the literal occurrences (see~\Cref{fig:overview-anchor-decoding} for an illustration). We interpret the symbols before \(Z\) as proposed true
occurrences and those after it as proposed false occurrences.
This partition need not be consistent. After optimizing the order
within each side, we show that the total LCS score depends on the sum of the clause score of \(S\) and the consistency
score of \(S'\). Thus, an optimal median must choose a partition
maximizing the sum of these two scores.

The partition induced by a satisfying assignment attains the maximum possible sum:
the clause score counts every clause, and the consistency score
counts exactly half of all occurrences. Conversely, attaining
both bounds yields a satisfying assignment. Thus, the optimum
median cost determines whether \(\varphi\) is satisfiable,
proving NP-hardness even with four input permutations.

In fact, the reduction is also \textit{gap-preserving}. An optimal
median corresponds to a partition, not necessarily
consistent, that maximizes the sum of the clause score on the
proposed true side and the consistency score on the proposed false
side. We show that inconsistent partitions never help: any extra
clause score obtained by placing incompatible literal occurrences
on the proposed true side is offset by a loss in the consistency
score. {The key is a quantitative decoding argument: for each variable, choose the value that makes as many of its occurrences in \(S\) true as possible, breaking ties arbitrarily, and let \(D\) count selected occurrences that remain false. The clause score can exceed the number of clauses satisfied by this assignment by at most \(D\), while balance makes the consistency score exactly \(N/2-D\), where \(N\) is the total number of literal-occurrence symbols. Thus, the decoded assignment induces a consistent partition with at least the same combined score, allowing the four input permutations to preserve the source gap. \Cref{fig:overview-anchor-decoding} summarizes how decoding bounds its combined score.} 

Combining this with the fact that \textsc{Max-3-SAT} is NP-hard to approximate to any factor better than 7/8 gives hardness
of approximation within every fixed factor smaller than \(53/52\),
even with exactly four input permutations. The details appear in
\Cref{sec:3sat-reduction}. 

\begin{figure}[!htbp]
\centering
\begin{tikzpicture}[
  font=\small,
  bar/.style={draw=black!55,minimum height=.55cm,inner sep=2pt},
  clause/.style={bar,fill=blue!10},
  consistency/.style={bar,fill=orange!15},
  anchorblock/.style={bar,fill=black!12},
  score/.style={draw=black!35,rounded corners=2pt,minimum height=.75cm,
    text width=5.9cm,align=center,inner sep=5pt}
]
  \node[font=\small\bfseries] at (3.1,.65) {Clause inputs};
  \node[font=\small\bfseries] at (10.1,.65) {Consistency inputs};
  \node[clause,minimum width=2cm] at (1.7,0) {$\beta_1$};
  \node[anchorblock,minimum width=2.8cm] at (4.1,0) {$Z$};
  \node[clause,minimum width=2cm] at (1.7,-.68) {$\beta_2$};
  \node[anchorblock,minimum width=2.8cm] at (4.1,-.68) {$Z$};
  \node[anchorblock,minimum width=2.8cm] at (9.1,0) {$Z$};
  \node[consistency,minimum width=2cm] at (11.5,0) {$\alpha_1$};
  \node[anchorblock,minimum width=2.8cm] at (9.1,-.68) {$Z$};
  \node[consistency,minimum width=2cm] at (11.5,-.68) {$\alpha_2$};

  \node[align=center,font=\scriptsize] at (6.6,-1.4)
    {Anchor lemma: some optimal median has the form};
  \draw[-{Stealth[length=2mm]},black!65] (6.6,-1.67)--(6.6,-1.96);
  \node[clause,minimum width=3cm] at (3.05,-2.3) {$\tau_S$};
  \node[anchorblock,minimum width=4.1cm] at (6.6,-2.3) {$Z$};
  \node[consistency,minimum width=3cm] at (10.15,-2.3) {$\tau_{S'}$};
  \node[font=\scriptsize] at (3.05,-2.82) {proposed true occurrences};
  \node[font=\scriptsize] at (10.15,-2.82) {proposed false occurrences};

  \node[font=\scriptsize,align=center] at (6.6,-3.43)
    {Decode each variable by its majority label in $S$};
  \node[score,fill=blue!5] (q) at (3.25,-4.16)
    {$q(S)\leq\Val_\varphi(f)+D$};
  \node[score,fill=orange!7] (k) at (9.95,-4.16)
    {$\kappa(S')=N/2-D$};
  \draw[-{Stealth[length=2mm]},black!65] (q.south)--(5.9,-4.96);
  \draw[-{Stealth[length=2mm]},black!65] (k.south)--(7.3,-4.96);
  \node[align=center] at (6.6,-5.3)
    {$q(S)+\kappa(S')\leq\Val_\varphi(f)+N/2$};
\end{tikzpicture}
\caption{From four input permutations to a partition, and from an
arbitrary partition to a consistent assignment. Here $N$ is the number
of literal-occurrence symbols, $q(S)$ is the clause score, and
$\kappa(S')$ is the consistency score. For the majority assignment $f$, with ties broken arbitrarily,
let $D$ count occurrences in $S$ that are false under $f$. Any extra clause score is at most $D$, while the
consistency score loses exactly $D$. Thus, inconsistent partitions
cannot improve the combined score.}
\label{fig:overview-anchor-decoding}
\end{figure}
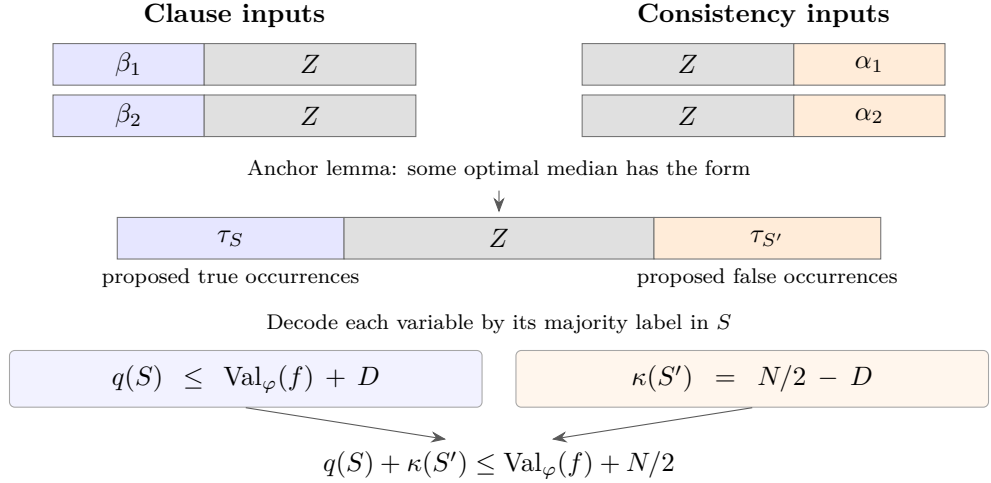

\paragraph{A general framework and stronger hardness.}
Looking back at the proof above, one may see that the only part of the
reduction specific to \textsc{3-SAT} is the pair of clause
permutations \(\beta_1,\beta_2\). It is therefore natural to ask
whether we can replace them by permutations for a different Boolean
constraint. We show that this is indeed possible. More precisely,
we associate each symbol with a variable and a Boolean label. For
\textsc{3-SAT}, the label records whether a literal occurrence is
positive or negative, or equivalently, the value that makes the
literal true. In general, the label should be interpreted as the
value that \textit{the symbol wants its associated variable to take}. It
then suffices to find, for each constraint, a pair of local
permutations whose score on the symbols agreeing with a proposed
local assignment is large exactly when the
constraint is satisfied. More precisely, each constraint must admit a polynomial-time constructible pair of permutations on \(B\) labeled symbols whose restricted LCS is exactly \(A+\Delta\) for every satisfying local assignment and exactly \(A\) for every unsatisfying one, for fixed constants \(B\), \(A\geq0\), and \(\Delta>0\) common to all constraints. The preceding argument -- that any gain in the constraint score from inconsistent choices is offset by a loss in the consistency score -- extends beyond \textsc{3-SAT}. Its key ingredient is gadget-independent: deleting the same \(D\) symbols from both restricted permutations decreases their LCS by at most \(D\), regardless of the internal structure of the local constraint gadgets. Assuming that each variable has equally many symbols of each label overall, the consistency and anchor gadgets can be reused, and the same argument controls inconsistent
partitions. The resulting hardness factor obtained from the reduction depends both on the gap of the starting CSP and on the quality of these local gadgets. 

We instantiate this framework with
\textsc{Max-E3-Lin-2}.\footnote{The constraints are linear equations
over \(\mathbb F_2\), each involving three distinct Boolean
variables; the goal is to maximize the number of satisfied
equations.}
We first construct a pair of balanced eight-symbol permutations \(P_j\) and \(P'_j\) whose restricted LCS
is three when the equation is satisfied and two otherwise. For an equation \(x\oplus y\oplus z=1\), the gadget in~\Cref{fig:overview-parity-gadget} uses two copies of each label for \(x\) and one copy of each label for \(y\) and \(z\), so every assignment selects four symbols and the gadget is balanced for each variable. The two selected \(x\)-symbols appear on either side of the selected \(y\)- and \(z\)-symbols in one permutation and are consecutive in the other, preventing all four symbols from forming a common subsequence. Their placement encodes parity: a satisfying assignment puts the selected \(y\)- and \(z\)-symbols on the same side of the consecutive \(x\)-pair, in the same order in both permutations, so they can be matched together with one \(x\)-symbol. An unsatisfying assignment produces two disjoint reversed pairs, forcing every common subsequence to omit at least one symbol from each pair and limiting the LCS to two. Flipping the labels of the \(x\)-symbols handles \(x\oplus y\oplus z=0\), interchanging the parity classes while preserving balance. We then use the fact that \textsc{Max-E3-Lin-2} is hard to
approximate beyond \(1/2+\varepsilon\), which is substantially stronger
than the \(7/8+\varepsilon\) guarantee for \textsc{3-SAT}.
Combining the gadget with this larger source gap improves the
hardness bound to every fixed factor smaller than \(35/34\).
The framework and its instantiation appear in
\Cref{sec:boolean-csp-framework,sec:e3lin-gadget}.

\begin{figure}[!htbp]
\centering
\begin{tikzpicture}[
  x=1.25cm,y=1cm,font=\small,
  symbol/.style={rounded corners=1pt,minimum width=0.90cm,
    minimum height=0.58cm,inner sep=2pt},
  excluded/.style={symbol,draw=black!45,dashed,
    text=black!65,fill=black!4},
  selected/.style={symbol,draw=black!70,text=black,fill=white},
  matched/.style={selected,draw=teal!70!black,fill=teal!12},
  match/.style={draw=teal!70!black,line width=0.9pt},
  rowlabel/.style={anchor=east},
  legend/.style={minimum width=0.38cm,minimum height=0.28cm}
]
  \node at (4.5,0.65)
    {\textbf{Unsatisfied:} $(x,y,z)=(0,0,0)$};
  \node[rowlabel] at (0.35,0) {$P_j$};
  \node[rowlabel] at (0.35,-1.2) {$P'_j$};

  \foreach \i/\s/\state in {
    1/{x_0^{(1)}}/selected, 2/{y_0}/matched,
    3/{z_0}/matched,       4/{x_1^{(1)}}/excluded,
    5/{x_1^{(2)}}/excluded,6/{z_1}/excluded,
    7/{y_1}/excluded,      8/{x_0^{(2)}}/selected}
    \node[\state] (fP\i) at (\i,0) {$\s$};

  \foreach \i/\s/\state in {
    1/{x_1^{(1)}}/excluded,2/{y_0}/matched,
    3/{z_1}/excluded,      4/{x_0^{(1)}}/selected,
    5/{x_0^{(2)}}/selected,6/{z_0}/matched,
    7/{y_1}/excluded,      8/{x_1^{(2)}}/excluded}
    \node[\state] (fQ\i) at (\i,-1.2) {$\s$};

  \draw[match] (fP2.south) -- (fQ2.north);
  \draw[match] (fP3.south) -- (fQ6.north);
  \node at (4.5,-1.85) {$\LCS(P_j|_A,P'_j|_A)=2$};

  \draw[black!20] (0.35,-2.2) -- (8.45,-2.2);

  \begin{scope}[yshift=-3.3cm]
    \node at (4.5,0.65)
      {\textbf{Satisfied:} $(x,y,z)=(0,0,1)$};
    \node[rowlabel] at (0.35,0) {$P_j$};
    \node[rowlabel] at (0.35,-1.2) {$P'_j$};

    \foreach \i/\s/\state in {
      1/{x_0^{(1)}}/selected,2/{y_0}/matched,
      3/{z_0}/excluded,      4/{x_1^{(1)}}/excluded,
      5/{x_1^{(2)}}/excluded,6/{z_1}/matched,
      7/{y_1}/excluded,      8/{x_0^{(2)}}/matched}
      \node[\state] (tP\i) at (\i,0) {$\s$};

    \foreach \i/\s/\state in {
      1/{x_1^{(1)}}/excluded,2/{y_0}/matched,
      3/{z_1}/matched,       4/{x_0^{(1)}}/selected,
      5/{x_0^{(2)}}/matched, 6/{z_0}/excluded,
      7/{y_1}/excluded,      8/{x_1^{(2)}}/excluded}
      \node[\state] (tQ\i) at (\i,-1.2) {$\s$};

    \draw[match] (tP2.south) -- (tQ2.north);
    \draw[match] (tP6.south) -- (tQ3.north);
    \draw[match] (tP8.south) -- (tQ5.north);
    \node at (4.5,-1.85) {$\LCS(P_j|_A,P'_j|_A)=3$};
  \end{scope}

  \node[excluded,legend] at (0.8,-5.95) {};
  \node[anchor=west,font=\scriptsize] at (1.05,-5.95)
    {not selected};
  \node[selected,legend] at (3.3,-5.95) {};
  \node[anchor=west,font=\scriptsize] at (3.55,-5.95)
    {selected, outside shown LCS};
  \node[matched,legend] at (7.0,-5.95) {};
  \node[anchor=west,font=\scriptsize] at (7.25,-5.95)
    {shown LCS};
\end{tikzpicture}
\caption{The eight-symbol parity gadget for $x\oplus y\oplus z=1$.
For the indicated local assignment $a$, the four symbols in
$A$ are the ones whose labels (indicated as subscripts) agree with \(a\) and shown with solid borders; the other four disagreeing symbols are gray
and dashed. Teal nodes and connecting lines exhibit a longest common
subsequence of the restrictions to $A$: $y_0z_0$ in the unsatisfied
example and $y_0z_1x_0^{(2)}$ in the satisfied example.
In general, the restricted LCS has length $2$ for every unsatisfying
assignment and length $3$ for every satisfying assignment.}
\label{fig:overview-parity-gadget}
\end{figure}
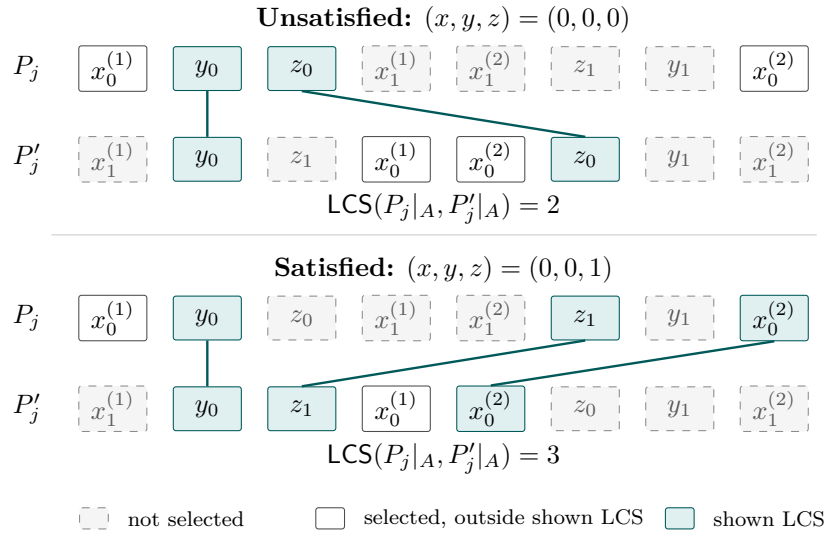

\paragraph{Median hardness to center hardness.} To transfer the hardness results to Ulam center, we use the cyclic block reduction from median to center for the Kendall-tau metric~\cite{biedl2009complexity}, whose applicability to the Ulam metric was observed in~\cite{CGJ21}. These works focus on polynomial-time reductions sufficient for establishing $\mathrm{NP}$-hardness, rather than explicitly formulating the approximation-preservation guarantees needed here. Given a Ulam median instance with $k$ input permutations over $d$ symbols, we create $k$ disjoint copies of its alphabet and construct $k$ center inputs over $kd$ symbols. All inputs use the same order of the alphabet copies, while the original permutations are arranged cyclically across them. Thus, every center input contains one copy of each original permutation, and every original permutation appears once in each block across the center instance.

Placing the same median candidate in every block makes its distance to each center input equal to its total distance in the original median instance. Conversely, the common block order lets us rearrange any center candidate blockwise without increasing its radius. Averaging over the center inputs then shows that at least one block defines a median candidate of cost no larger than that radius. These observations establish the exact preservation of the optimum value and the lossless extraction of a median solution from any center solution (see~\Cref{lem:median-to-center}). Together with the preservation of the number of inputs and the alphabet size of $kd$ in the center instance, these guarantees transfer both multiplicative and additive hardness from the Ulam median to the Ulam center. For completeness, we provide a detailed proof of these guarantees in~\Cref{sec:reduction-med-center}.

\section{Preliminaries}\label{sec:preliminaries}
\noindent \textbf{Notations. } For $m\in\mathbb N$, we use $[m]$ to denote the set $\{1,\ldots,m\}$. We use the notation $\sqcup$ to denote the union of two disjoint sets. For any set $X$ partitioned into $t$ disjoint subsets $X_1,X_2,\ldots,X_t$, we thus write $X = X_1 \sqcup X_ 2 \sqcup \ldots \sqcup X_t$.

\paragraph{Permutations and Ulam distance.} For an alphabet $\Sigma$, let $\mathcal P(\Sigma)$ denote the set of all permutations of $\Sigma$. For a permutation $\pi\in\mathcal P(\Sigma)$ and a subset $S\subseteq\Sigma$, let $\pi|_S$ denote the restriction of $\pi$ to the symbols in $S$, with their relative order preserved. For two sequences $x$ and $y$, let $\LCS(x,y)$ denote the length of a \emph{longest common subsequence}. For any two $x,y\in\mathcal P(\Sigma)$, their \emph{Ulam distance} is
\[
    \U(x,y):=|\Sigma|-\LCS(x,y).
\]

We will repeatedly use the following elementary identity, which follows from the triangle inequality for the Ulam distance.

\begin{lemma}\label{lem:two-permutation-score}
Let $\Sigma$ be an alphabet and let $\pi,\sigma\in\mathcal P(\Sigma)$. Then
\[
    \max_{\tau\in\mathcal P(\Sigma)}\bigl(\LCS(\tau,\pi)+\LCS(\tau,\sigma)\bigr)=|\Sigma|+\LCS(\pi,\sigma).
\]
\end{lemma}
\begin{proof}
Let $d:=|\Sigma|$.  For every $\tau\in\mathcal P(\Sigma)$, the triangle
inequality for Ulam distance gives
\[
 \U(\tau,\pi)+\U(\tau,\sigma)\geq \U(\pi,\sigma).
\]
Using the definition of the Ulam distance between any two permutations $\rho,\rho'$, $\U(\rho,\rho')=d-\LCS(\rho,\rho')$, and rearranging, we obtain
\[
 \LCS(\tau,\pi)+\LCS(\tau,\sigma)
 \leq d+\LCS(\pi,\sigma).
\]
This upper bound is attained by taking $\tau=\pi$, since
$\LCS(\pi,\pi)=d$.  The claimed equality follows.
\end{proof}

\paragraph{Ulam median.} An instance of \emph{Ulam median} is a (multi)-set $\Pi=\{\pi_1,\ldots,\pi_k\}$ of permutations of a common alphabet $\Sigma$. The problem asks to find a permutation $\sigma\in\mathcal P(\Sigma)$ that minimizes the total distance to input permutations, namely, $\cost_\Pi(\sigma):=\sum_{i=1}^k \U(\sigma,\pi_i)$. We denote the optimum value of the instance $\Pi$ by
\[
    \OPT_{\mathrm{med}}(\Pi):=\min_{\sigma\in\mathcal P(\Sigma)}\cost_\Pi(\sigma).
\]

\paragraph{Ulam center.} An instance of \emph{Ulam center} is a (multi)-set $\Pi=\{\pi_1,\ldots,\pi_k\}$ of permutations of a common alphabet $\Sigma$. The problem asks to find a permutation $\sigma\in\mathcal P(\Sigma)$ that minimizes the radius $\rad_\Pi(\sigma):=\max_{i\in[k]} \U(\sigma,\pi_i)$. We denote the optimum value of the instance $\Pi$ by
\[
    \OPT_{\mathrm{ctr}}(\Pi):=\min_{\sigma\in\mathcal P(\Sigma)}\rad_\Pi(\sigma).
\]

\paragraph{Approximation scheme.} For the Ulam median problem, an algorithm is a polynomial-time approximation scheme (PTAS) if, given any constant $\varepsilon > 0$ and an instance $\Pi$ of $k$ permutations over an alphabet of size $d$, it computes in polynomial time a solution $\rho$ satisfying $\cost_{\Pi}(\rho) \leq (1+\varepsilon)\OPT_{\mathrm{med}}(\Pi)$. An algorithm is an $\varepsilon kd$-additive approximation scheme for the Ulam median problem if it outputs a solution $\rho$ in polynomial time such that $\cost_{\Pi}(\rho) \leq \OPT_{\mathrm{med}}(\Pi) + \varepsilon kd$. The definition of PTAS for Ulam center is analogous to that of the median problem. An algorithm is an $\varepsilon d$-additive approximation scheme for the Ulam center problem if it outputs a solution $\tau$ in polynomial time such that $\rad_{\Pi}(\tau) \leq \OPT_{\mathrm{ctr}}(\Pi) + \varepsilon d$.

\paragraph{\textsc{Max-E3-Lin-2}.}
An instance of \textsc{Max-E3-Lin-2} is a Boolean CSP
\(\varphi=(V,\mathcal C)\), where
\(V=\{x_1,\ldots,x_n\}\) is a set of Boolean variables and
\(\mathcal C=(C_1,\ldots,C_m)\) is a list of equations over
\(\mathbb F_2\). Each equation has the form
\[
C_j:\qquad
x_{j,1}\oplus x_{j,2}\oplus x_{j,3}=b_j,
\]
where \(x_{j,1},x_{j,2},x_{j,3}\in V\) are distinct variables and
\(b_j\in\{0,1\}\). For an assignment \(f:V\to\{0,1\}\), let
\[
\Val_\varphi(f)
:=
\bigl|
\{j\in[m]:f\text{ satisfies }C_j\}
\bigr|,
\]
and define
\[
\OPT_{\mathrm{CSP}}(\varphi)
:=
\max_{f:V\to\{0,1\}}\Val_\varphi(f).
\]

We use the following gap form of H{\aa}stad's hardness theorem.

\begin{theorem}[{H{\aa}stad~\cite[Theorem~5.4]{Hastad01}}]
\label{thm:hastad-e3lin-gap}
For every fixed \(0<\eta<1/4\), it is \(\mathrm{NP}\)-hard to
distinguish instances \(\varphi\) of \textsc{Max-E3-Lin-2} with
\(m\) equations between the following two cases:
\begin{itemize}
    \item \emph{(Completeness)}
    \(
    \OPT_{\mathrm{CSP}}(\varphi)\geq(1-\eta)m.
    \)
    \item \emph{(Soundness)}
    \(
    \OPT_{\mathrm{CSP}}(\varphi)
    \leq\left(\frac12+\eta\right)m.
    \)
\end{itemize}
\end{theorem}

\section{The Basic \textsc{3-SAT} Reduction}
\label{sec:3sat-reduction}

In this section, we give a direct reduction from \textsc{3-SAT} to
Ulam median on four input permutations. Let
\[
\varphi:=C_1\wedge C_2\wedge\cdots\wedge C_m
\]
be a \(3\)-CNF formula on the variable set
\(V=\{x_1,x_2,\ldots,x_n\}\). We assume that every clause contains
exactly three literal occurrences, and that, for every variable
\(x\in V\), the number of positive occurrences of \(x\) equals the
number of negative occurrences of \(x\).\footnote{These assumptions do not affect our results; see~\Cref{app:balanced-3sat-hardness}.}
For an assignment \(f:V\to\{0,1\}\) to the variables, let
\[
\Val_\varphi(f)
:=
\bigl|
\{j\in[m]:f\text{ satisfies }C_j\}
\bigr|
\]
denote the number of clauses satisfied by \(f\), and define
\[
\OPT_{\mathrm{Max3SAT}}(\varphi)
:=
\max_{f:V\to\{0,1\}}\Val_\varphi(f).
\]

Our construction will contain one symbol for each \emph{literal
occurrence} in \(\varphi\). In particular, two appearances of the
same literal will correspond to distinct symbols. Let \(\Sigma\)
denote the set of all such literal-occurrence symbols, and let
\(
N:=|\Sigma|=3m.
\)

For a symbol \(\sigma\in\Sigma\), let
\(\operatorname{var}(\sigma)\) denote the variable underlying that
literal occurrence. We associate with \(\sigma\) a
\emph{label} \(\ell(\sigma)\in\{0,1\}\), defined to be the value of
\(\operatorname{var}(\sigma)\) that makes the corresponding literal
true. In other words, a positive occurrence of \(x\) has label \(1\), while a
negative occurrence of \(x\) has label \(0\).

For every variable \(x\in V\) and \(b\in\{0,1\}\), define
\[
\Sigma_x^b
:=
\{\sigma\in\Sigma:
\operatorname{var}(\sigma)=x,\ \ell(\sigma)=b\}.
\]
Since every variable has an equal number of positive and negative occurrences, we have
\begin{equation}
\label{eq:3sat-balance}
|\Sigma_x^0|=|\Sigma_x^1|
\qquad
\text{for every }x\in V.
\end{equation}
For each variable \(x\), write
\(r_x:=|\Sigma_x^0|=|\Sigma_x^1|\), so that
\(\sum_{x\in V}r_x=N/2\).
For each clause \(C_j\), let
\(\Sigma(C_j)\subseteq\Sigma\) denote the set of its three literal-occurrence
symbols. A central concept in our discussion is that of a \textit{truth partition},
which we now define.

\paragraph{Truth partitions.}
We refer to any partition
\(\Sigma=T\sqcup F\) as a \emph{truth partition},
interpreting \(T\) and \(F\) as the proposed sets of true and false
literal occurrences. A truth partition is \emph{consistent} if, for
every variable \(x\in V\), there exists \(b\in\{0,1\}\) such that
\[
\Sigma_x^b\subseteq T
\qquad\text{and}\qquad
\Sigma_x^{1-b}\subseteq F.
\]
In other words, a consistent partition disallows placing occurrences
of the same literal on different sides, or occurrences of a literal
and its negation on the same side.

Next, we define the \textit{agreement set} corresponding to an assignment to the variables. For an assignment \(f:V\to\{0,1\}\) define
\[
\Agr(f)
:=
\{\sigma\in\Sigma:
f(\operatorname{var}(\sigma))=\ell(\sigma)\}.
\]
Note that \(\Agr(f)\) contains precisely the literal
occurrences made true by \(f\). Thus, \(f\) induces the consistent
truth partition
\[
T=\Agr(f),
\qquad
F=\Sigma\setminus\Agr(f).
\]
Under this partition, \(T\cap\Sigma(C_j)\neq\varnothing\) precisely
when \(f\) satisfies the clause \(C_j\). The balance condition also gives
\begin{equation}
\label{eq:3sat-agreement-size}
|\Agr(f)|=\frac N2
\qquad
\text{for every assignment }f.
\end{equation}

We are now ready to describe our construction. At a high level, our
construction consists of two pairs of permutations. The first pair,
which we refer to as the \emph{clause-witness gadget}, provides a clause score, which records how
many clauses are satisfied by a set of proposed true literals. The
second pair, which we refer to as the
\emph{variable-consistency gadget}, supplies a consistency score, which records how consistent different occurrences of the same literal are across clauses. For the remainder of this section, fix an arbitrary total ordering
of the symbols in \(\Sigma\).

\paragraph{Clause-witness gadget.}
For each clause \(C_j\), let \(A_j\in\Sigma^3\) be the string
obtained by listing the three symbols in \(\Sigma(C_j)\) according
to the fixed ordering of \(\Sigma\). Let \(A_j^{\mathcal R}\)
denote its reversal, and define
\[
\beta_1:=A_1A_2\cdots A_m,
\qquad
\beta_2:=A_1^{\mathcal R}A_2^{\mathcal R}\cdots A_m^{\mathcal R}.
\]
The basic property of this gadget is
\(\LCS(\beta_1,\beta_2)=m\): a common subsequence may contain at most
one symbol from each clause and may choose one symbol from every clause.
We will use the following stronger form where we restrict to an arbitrary subset of all literal occurrences.

\begin{lemma}
\label{lem:clause-gadget}
For every \(S\subseteq\Sigma\),
\[
\LCS(\beta_1|_S,\beta_2|_S)
=
\bigl|
\{j\in[m]:S\cap\Sigma(C_j)\neq\varnothing\}
\bigr|.
\]
\end{lemma}

\begin{proof}
For every \(j\in[m]\), the relative order of the symbols in
\(\Sigma(C_j)\) is reversed between \(\beta_1\) and \(\beta_2\).
Therefore, a common subsequence of \(\beta_1|_S\) and
\(\beta_2|_S\) contains at most one symbol from
\(S\cap\Sigma(C_j)\). Hence
\[
\LCS(\beta_1|_S,\beta_2|_S)
\leq
\bigl|
\{j\in[m]:S\cap\Sigma(C_j)\neq\varnothing\}
\bigr|.
\]

Conversely, for every \(j\in[m]\) such that
\(S\cap\Sigma(C_j)\neq\varnothing\), choose one symbol from this
intersection. Since the clause blocks occur in the same order in
\(\beta_1\) and \(\beta_2\), the chosen symbols, listed in increasing
order of \(j\), form a common subsequence. This proves the reverse
inequality and hence the claimed equality.
\end{proof}

\paragraph{Variable-consistency gadget.}
For every variable \(x\in V\) and \(b\in\{0,1\}\), let \(B_x^b\)
be the string obtained by listing the symbols in \(\Sigma_x^b\)
according to the fixed ordering of \(\Sigma\). Define
\[
\alpha_1
:=
B_{x_1}^0B_{x_1}^1
B_{x_2}^0B_{x_2}^1
\cdots
B_{x_n}^0B_{x_n}^1,
\]
\[
\alpha_2
:=
B_{x_1}^1B_{x_1}^0
B_{x_2}^1B_{x_2}^0
\cdots
B_{x_n}^1B_{x_n}^0.
\]
For each variable, the two label blocks appear in opposite orders in
\(\alpha_1\) and \(\alpha_2\). Consequently, a common subsequence
can use symbols from at most one label class of each variable.

\begin{lemma}
\label{lem:consistency-gadget}
For every \(S\subseteq\Sigma\),
\[
\LCS(\alpha_1|_S,\alpha_2|_S)
=
\sum_{x\in V}
\max\bigl\{
|S\cap\Sigma_x^0|,
|S\cap\Sigma_x^1|
\bigr\}.
\]
Consequently,
\(
\LCS(\alpha_1|_S,\alpha_2|_S)\leq N/2.
\)
\end{lemma}

\begin{proof}
Fix \(x\in V\). Every symbol of \(B_x^0\) precedes every symbol of
\(B_x^1\) in \(\alpha_1\), while every symbol of \(B_x^1\) precedes
every symbol of \(B_x^0\) in \(\alpha_2\). Therefore, a common
subsequence contains symbols from at most one of these two blocks.
Its contribution from the blocks for \(x\) is consequently at most
\(
\max\bigl\{
|S\cap\Sigma_x^0|,
|S\cap\Sigma_x^1|
\bigr\}.
\)

This bound can be attained simultaneously for every variable: for
each \(x\), choose the larger of \(S\cap\Sigma_x^0\) and
\(S\cap\Sigma_x^1\), list its symbols in their common internal
order, and concatenate these choices in variable order. This proves
the formula. The upper bound follows from
\[
\LCS(\alpha_1|_S,\alpha_2|_S)
\leq
\sum_{x\in V}r_x
=
\frac N2.
\]
\end{proof}

\paragraph{Adding an anchor block.}
Let \(\Sigma_Z:=\{z_1,z_2,\ldots,z_M\}\) be an alphabet disjoint
from \(\Sigma\), where \(M:=2N+1\), and define
\[
Z:=z_1z_2\cdots z_M.
\]
We refer to \(Z\) as the \emph{anchor block} and its symbols as
\emph{anchor symbols}. Finally, construct the four input permutations
\[
L_1:=\beta_1Z,
\qquad
L_2:=\beta_2Z,
\qquad
R_1:=Z\alpha_1,
\qquad
R_2:=Z\alpha_2,
\]
and let
\[
\Pi_\varphi:=\{L_1,L_2,R_1,R_2\}.
\]

The following lemma shows that the anchor block acts as a separator:
every candidate median may be transformed, without decreasing any of
its four LCS scores, into one in which \(Z\) appears as a contiguous
block.

\begin{lemma}[Anchor-block lemma]
\label{lem:anchor-block}
Let \(\gamma_1,\gamma_2,\delta_1,\delta_2\) be arbitrary
permutations of\/ \(\Sigma\). For every permutation \(\pi\) of\/
\(\Sigma\sqcup \Sigma_Z\), there exist a partition
\(\Sigma=S_L\sqcup S_R\) and permutations
\(\tau_L,\tau_R\) of \(S_L,S_R\), respectively, such that
\(\pi^\star:=\tau_LZ\tau_R\) satisfies
\[
\LCS(\pi^\star,\gamma_iZ)
\geq
\LCS(\pi,\gamma_iZ),
\]
\[
\LCS(\pi^\star,Z\delta_i)
\geq
\LCS(\pi,Z\delta_i)
\]
for each \(i\in\{1,2\}\).
\end{lemma}

We defer the proof of \Cref{lem:anchor-block} to
\Cref{sec:anchor-block-proof}. Note that this lemma is independent of the \textsc{3-SAT} instance
\(\varphi\): it applies to any four permutations of \(\Sigma\).
It allows us to restrict attention, without loss of generality, to
only candidate medians of the form \(\tau_LZ\tau_R\). Each such
``canonical'' permutation induces a truth partition
\(\Sigma=S_L\sqcup S_R\), where \(S_L\) and \(S_R\)
are the symbols placed before and after \(Z\), respectively. We
interpret \(S_L\) as the proposed set of true literal occurrences
and \(S_R\) as the proposed set of false literal occurrences. Note that this
truth partition need not be consistent. Once the partition is fixed,
the only remaining choice is the ordering of the symbols within
\(S_L\) and \(S_R\).

\begin{lemma}[Fixed-partition optimum]
\label{lem:canonical-score}
Fix a partition
\(\Sigma=S_L\sqcup S_R\). Then
\[
\begin{aligned}
&\max_{\tau_L,\tau_R}
\sum_{i=1}^2
\Bigl(
\LCS(\tau_LZ\tau_R,L_i)
+\LCS(\tau_LZ\tau_R,R_i)
\Bigr)\\
&\qquad=
4M+N
+\LCS(\beta_1|_{S_L},\beta_2|_{S_L})
+\LCS(\alpha_1|_{S_R},\alpha_2|_{S_R}),
\end{aligned}
\]
where the maximum ranges over all permutations \(\tau_L\) of
\(S_L\) and \(\tau_R\) of \(S_R\).
\end{lemma}

\begin{proof}
Let \(\pi=\tau_LZ\tau_R\). Since \(M>N\), a longest common
subsequence of \(\pi\) and \(L_i=\beta_iZ\) cannot contain a symbol
from \(\tau_R\): any such subsequence contains no anchor symbol and
therefore has length at most \(N<M\), whereas \(Z\) itself is a
common subsequence of length \(M\). Hence
\[
\LCS(\pi,L_i)
=
M+\LCS(\tau_L,\beta_i|_{S_L}).
\]
Symmetrically,
\[
\LCS(\pi,R_i)
=
M+\LCS(\tau_R,\alpha_i|_{S_R}).
\]

Applying \Cref{lem:two-permutation-score} independently on the two
sides gives
\[
\max_{\tau_L}
\sum_{i=1}^2
\LCS(\tau_L,\beta_i|_{S_L})
=
|S_L|+\LCS(\beta_1|_{S_L},\beta_2|_{S_L}),
\]
\[
\max_{\tau_R}
\sum_{i=1}^2
\LCS(\tau_R,\alpha_i|_{S_R})
=
|S_R|+\LCS(\alpha_1|_{S_R},\alpha_2|_{S_R}).
\]
These maxima are attained by
\(\tau_L=\beta_1|_{S_L}\) and
\(\tau_R=\alpha_1|_{S_R}\), respectively.
Adding the equalities and using \(|S_L|+|S_R|=N\) proves the lemma.
\end{proof}

For \(S\subseteq\Sigma\), the expressions \(\LCS(\beta_1|_S,\beta_2|_S)\) and \(\LCS(\alpha_1|_S,\alpha_2|_S)\) are important enough to give them names. We call them the \textit{clause score} and \textit{consistency score}, respectively, and  write
\[
q(S):=\LCS(\beta_1|_S,\beta_2|_S),
\qquad
\kappa(S):=\LCS(\alpha_1|_S,\alpha_2|_S).
\]
The names are suggestive, as \(q(S)\) counts the clauses containing
an occurrence from \(S\), while \(\kappa(S)\) is the maximum number
of occurrences in \(S\) that can simultaneously evaluate to false
under a single assignment.

Now, since every
input permutation has length \(M+N\),
\Cref{lem:anchor-block,lem:canonical-score} give
\begin{equation}
\label{eq:3sat-partition-optimum}
\OPT_{\mathrm{med}}(\Pi_\varphi)
=
3N-
\max_{\Sigma=S_L\sqcup S_R}
\bigl(q(S_L)+\kappa(S_R)\bigr).
\end{equation}
In other words, an optimal median must correspond to a truth partition, not necessarily consistent, that maximizes the sum of the clause score of the proposed true side and the consistency score of the proposed false side. We will show that inconsistent partitions never help---any extra clause
score obtained by placing incompatible literal
occurrences on the proposed true side is offset by a loss in the consistency
score.

We are now ready to prove the main guarantee of our construction,

\begin{theorem}
\label{thm:3sat-exact-relation}
The instance \(\Pi_\varphi\) satisfies
\[
\OPT_{\mathrm{med}}(\Pi_\varphi)
=
\frac{5N}{2}
-\OPT_{\mathrm{Max3SAT}}(\varphi)
=
\frac{15m}{2}
-\OPT_{\mathrm{Max3SAT}}(\varphi).
\]
\end{theorem}

\begin{proof}
By \eqref{eq:3sat-partition-optimum}, it suffices to prove that
\[
\max_{\Sigma=S_L\sqcup S_R}
\bigl(q(S_L)+\kappa(S_R)\bigr)
=
\frac N2+\OPT_{\mathrm{Max3SAT}}(\varphi).
\]
First, let \(f\) be an assignment satisfying
\(\OPT_{\mathrm{Max3SAT}}(\varphi)\) clauses, and consider its
induced truth partition
\[
T=\Agr(f),
\qquad
F=\Sigma\setminus\Agr(f).
\]
By \Cref{lem:clause-gadget,lem:consistency-gadget} and the balance
condition,
\[
q(T)=\Val_\varphi(f)=\OPT_{\mathrm{Max3SAT}}(\varphi),
\qquad
\kappa(F)=\frac N2.
\]
This proves \(
\max_{\Sigma=S_L\sqcup S_R}
\bigl(q(S_L)+\kappa(S_R)\bigr)
\geq
\frac N2+\OPT_{\mathrm{Max3SAT}}(\varphi).
\)

For the reverse inequality, fix an arbitrary truth partition
\(\Sigma=S_L\sqcup S_R\). The symbols in \(S_L\) are
being treated as true, but they need not be simultaneously true
under any assignment. We therefore choose an assignment that makes
as many of these proposed true occurrences as possible actually
true. This can be done independently for each variable: choose
\[
f(x)\in\arg\max_{b\in\{0,1\}}|S_L\cap\Sigma_x^b|,
\]
breaking ties arbitrarily. Let
\[
D:=|S_L\setminus\Agr(f)|.
\]
Thus, \(D\) counts the occurrences placed on the proposed true side
that nevertheless evaluate to false under \(f\).

We first compare the clause score with the number of clauses
satisfied by \(f\). By \Cref{lem:clause-gadget},
\[
q(S_L)
=
\bigl|
\{j\in[m]:S_L\cap\Sigma(C_j)\neq\varnothing\}
\bigr|.
\]
A clause counted by \(q(S_L)\) but not satisfied by \(f\) must
contain one of the \(D\) false occurrences in \(S_L\). Since each
literal occurrence belongs to exactly one clause, there are at most
\(D\) such clauses. Therefore,
\begin{equation}
\label{eq:3sat-clause-decoding}
q(S_L)\leq\Val_\varphi(f)+D.
\end{equation}
In other words, the proposed partition can gain at most \(D\) clause
witnesses beyond those justified by the decoded assignment.

We now show that these disagreements cost exactly \(D\) in the
consistency score. Fix a variable \(x\). By the choice of \(f(x)\) and \eqref{eq:3sat-balance},
\[
\begin{aligned}
|S_R\cap\Sigma_x^{1-f(x)}|
&=r_x-|S_L\cap\Sigma_x^{1-f(x)}|\\
&\geq r_x-|S_L\cap\Sigma_x^{f(x)}|\\
&=|S_R\cap\Sigma_x^{f(x)}|.
\end{aligned}
\] Thus,
\(\Sigma_x^{1-f(x)}\) is a largest label class inside \(S_R\).
By \Cref{lem:consistency-gadget},
\[
\kappa(S_R)
=
\sum_{x\in V}|S_R\cap\Sigma_x^{1-f(x)}|.
\]
This sum counts precisely the occurrences in \(S_R\) that evaluate
to false under \(f\). By balance, exactly \(N/2\) literal
occurrences evaluate to false in total. Of these, exactly \(D\)
lie in \(S_L\), and all the others lie in \(S_R\). Hence
\begin{equation}
\label{eq:3sat-consistency-decoding}
\kappa(S_R)=\frac N2-D.
\end{equation}
Combining
\Cref{eq:3sat-clause-decoding,eq:3sat-consistency-decoding}, we obtain
\begin{align*}
q(S_L)+\kappa(S_R)
&\leq
\Val_\varphi(f)+D+\frac N2-D\\
&=
\Val_\varphi(f)+\frac N2\\
&\leq
\OPT_{\mathrm{Max3SAT}}(\varphi)+\frac N2.
\end{align*}
Thus, the consistent partition induced by \(f\) has combined score
at least as large as that of the original partition: any additional
clause witnesses are paid for by the loss in the consistency score.
This proves the required upper bound. Applying
\eqref{eq:3sat-partition-optimum} and using \(N=3m\) completes the
proof.
\end{proof}
In particular, \Cref{thm:3sat-exact-relation} gives
\[
\varphi\text{ is satisfiable}
\quad\Longleftrightarrow\quad
\OPT_{\mathrm{med}}(\Pi_\varphi)\leq\frac{13m}{2}.
\]
Each input permutation has length
\(M+N=3N+1=9m+1\), and the construction can be carried out in
polynomial time. This completes the reduction from
\textsc{3-SAT}. \Cref{lem:balanced-3sat-decision} immediately yields that Ulam median is NP-hard already at four input permutations. 

\paragraph{Approximation hardness.}
\Cref{thm:3sat-exact-relation} further allows us
to transfer approximation hardness for \textsc{Max-3-SAT} to Ulam
median. We use the balanced gap result established in
\Cref{lem:balanced-3sat-gap}: for every fixed
\(0<\eta<1/16\), it is $\mathrm{NP}$-hard to distinguish formulas satisfying
the assumptions above for which
\(
\OPT_{\mathrm{Max3SAT}}(\varphi)\geq(1-\eta)m
\)
from those for which
\(
\OPT_{\mathrm{Max3SAT}}(\varphi)
\leq\left(\frac78+\eta\right)m.
\)
Applying \Cref{thm:3sat-exact-relation} gives the following gap for
Ulam median.

\begin{corollary}
\label{cor:ulam-gap-3sat}
For every fixed \(0<\eta<1/16\), given a four-permutation Ulam-median instance \(\Pi\), in which each
input permutation has length \(9m+1\), it is $\mathrm{NP}$-hard to distinguish
between the following two cases.
\begin{itemize}
    \item \emph{(Completeness)} \(
\OPT_{\mathrm{med}}(\Pi)
\leq\left(\frac{13}{2}+\eta\right)m
\)
\item \emph{(Soundness)} \(
\OPT_{\mathrm{med}}(\Pi)
\geq\left(\frac{53}{8}-\eta\right)m.
\)
\end{itemize}
\end{corollary}

\begin{proof}
In the completeness case, \Cref{thm:3sat-exact-relation} gives
\[
\OPT_{\mathrm{med}}(\Pi_\varphi)
=\frac{15m}{2}-\OPT_{\mathrm{Max3SAT}}(\varphi)
\leq\left(\frac{13}{2}+\eta\right)m.
\]
In the soundness case, the same theorem gives
\[
\OPT_{\mathrm{med}}(\Pi_\varphi)
\geq\frac{15m}{2}-\left(\frac78+\eta\right)m
=\left(\frac{53}{8}-\eta\right)m.
\]
\end{proof}

Note that the ratio between the soundness and completeness thresholds
approaches
\(
\frac{53/8}{13/2}=\frac{53}{52}
\)
as \(\eta\to0\). We thus obtain the following consequence.

\begin{corollary}
\label{cor:53-52-hardness}
For every fixed \(0<\varepsilon<1/52\), it is $\mathrm{NP}$-hard to
approximate Ulam median within a factor of
\(53/52-\varepsilon\), even when the input consists of exactly
four permutations.
\end{corollary}

\begin{proof}
Fix \(0<\varepsilon<1/52\), and set
\(\rho:=53/52-\varepsilon\). Choose \(0<\eta<1/16\) sufficiently
small that
\[
\rho\left(\frac{13}{2}+\eta\right)
<\frac{53}{8}-\eta.
\]
On a completeness instance from \Cref{cor:ulam-gap-3sat}, a
\(\rho\)-approximation algorithm would return a permutation of
cost strictly less than \((53/8-\eta)m\). On a soundness instance,
every permutation has cost at least \((53/8-\eta)m\). Such an
algorithm would therefore distinguish the two cases, proving the
claim.
\end{proof}

\subsection{Proof of the Anchor-Block Lemma}
\label{sec:anchor-block-proof}

In this section, we give a proof of the anchor block lemma.

\begin{proof}[Proof of~\Cref{lem:anchor-block}]
Fix an arbitrary permutation \(\pi\) of
\(\Sigma\sqcup \Sigma_Z\). We transform \(\pi\) into the
desired canonical form in two steps.

\paragraph{Step 1: sorting the anchor symbols.}
Let \(\widehat{\pi}\) be obtained from \(\pi\) by permuting only the
symbols of \(\Sigma_Z\) so that they occur in the order
\(z_1,z_2,\ldots,z_M\). The symbols of \(\Sigma\) remain in their
original locations and hence retain their relative order. We claim
that, for each \(i\in\{1,2\}\),
\[
\LCS(\widehat{\pi},\gamma_iZ)
\geq
\LCS(\pi,\gamma_iZ)
\quad\text{and}\quad
\LCS(\widehat{\pi},Z\delta_i)
\geq
\LCS(\pi,Z\delta_i).
\]

For \(x\in\Sigma\), let \(a(x)\) denote the number of symbols from
\(\Sigma_Z\) that precede \(x\) in \(\pi\). Since the reordering does
not change the locations occupied by symbols of \(\Sigma_Z\), the
anchor symbols preceding \(x\) in \(\widehat{\pi}\) are precisely
\(z_1,\ldots,z_{a(x)}\), while those following \(x\) are precisely
\(z_{a(x)+1},\ldots,z_M\).

Fix \(i\in\{1,2\}\), and let \(\rho\) be a longest common subsequence
of \(\pi\) and \(\gamma_iZ\). If \(\rho\) contains no symbol of
\(\Sigma\), then \(|\rho|\leq M\), whereas \(Z\) is a common
subsequence of \(\widehat{\pi}\) and \(\gamma_iZ\) of length \(M\).

Otherwise, let \(x\) be the last symbol of \(\Sigma\) appearing in
\(\rho\). Since every symbol of \(\Sigma\) precedes every anchor symbol
in \(\gamma_iZ\), all anchor symbols appearing in \(\rho\) must follow
\(x\) in \(\pi\). Hence \(\rho\) contains at most \(M-a(x)\) anchor
symbols. Moreover,
\[
\bigl(\rho|_\Sigma\bigr)
z_{a(x)+1}z_{a(x)+2}\cdots z_M
\]
is a common subsequence of \(\widehat{\pi}\) and \(\gamma_iZ\). Its
length is at least \(|\rho|\), and therefore
\[
\LCS(\widehat{\pi},\gamma_iZ)
\geq
\LCS(\pi,\gamma_iZ).
\]
The proof for \(Z\delta_i\) is symmetric. Thus,
\[
\LCS(\widehat{\pi},Z\delta_i)
\geq
\LCS(\pi,Z\delta_i).
\]
It therefore suffices to consider the case in which the anchor symbols
occur in the order \(z_1,z_2,\ldots,z_M\). We may then write
\[
\widehat{\pi}
=
U_0z_1U_1z_2\cdots z_MU_M,
\]
where each \(U_j\) is a possibly empty string over \(\Sigma\).

\paragraph{Step 2a: analyzing \(\gamma_iZ\).}
Fix \(i\in\{1,2\}\). We claim that
\begin{equation}
\LCS(\widehat{\pi},\gamma_iZ)
=
M+\max_{0\leq t\leq M}
\bigl(
\LCS(U_0U_1\cdots U_t,\gamma_i)-t
\bigr).
\label{eq:cut-1}
\end{equation}

Indeed, for any \(t\in\{0,\ldots,M\}\), a common subsequence of
\(U_0U_1\cdots U_t\) and \(\gamma_i\) can be followed by
\(z_{t+1}z_{t+2}\cdots z_M\). Hence
\[
\LCS(\widehat{\pi},\gamma_iZ)
\geq
\LCS(U_0U_1\cdots U_t,\gamma_i)+M-t.
\]

For the reverse inequality, let \(\rho\) be any common subsequence of
\(\widehat{\pi}\) and \(\gamma_iZ\). If \(\rho\) contains no symbol
from \(\Sigma\), then \(|\rho|\leq M\), which is bounded by the
right-hand side of~\Cref{eq:cut-1} by taking \(t=0\).

Otherwise, suppose that the last symbol of \(\Sigma\) appearing in
\(\rho\) belongs to \(U_t\). Then \(\rho|_\Sigma\) is a common
subsequence of \(U_0U_1\cdots U_t\) and \(\gamma_i\). Moreover, every
symbol of \(\Sigma_Z\) appearing in \(\rho\) must belong to
\(z_{t+1}z_{t+2}\cdots z_M\), so \(\rho\) contains at most \(M-t\)
such symbols. Therefore,
\[
|\rho|
\leq
\LCS(U_0U_1\cdots U_t,\gamma_i)+M-t.
\]
This proves~\Cref{eq:cut-1}.

Since \(\gamma_i\) has length \(N\), we have
\(\LCS(U_0U_1\cdots U_t,\gamma_i)\leq N\). Thus, when \(t>N\), the
quantity inside the maximum in~\Cref{eq:cut-1} is negative, whereas
for \(t=0\) it is nonnegative. Consequently, the maximum is attained
for some \(t\leq N\).

It follows that, for each \(i\in\{1,2\}\), there exists a longest
common subsequence of \(\widehat{\pi}\) and \(\gamma_iZ\) whose symbols
from \(\Sigma\) all belong to \(U_0U_1\cdots U_N\).

\paragraph{Step 2b: analyzing \(Z\delta_i\).}
A symmetric argument gives
\begin{equation}
\LCS(\widehat{\pi},Z\delta_i)
=
M+\max_{0\leq t\leq M}
\bigl(
\LCS(U_tU_{t+1}\cdots U_M,\delta_i)-(M-t)
\bigr).
\label{eq:cut-2}
\end{equation}

Since \(\delta_i\) has length \(N\), we have
\(\LCS(U_tU_{t+1}\cdots U_M,\delta_i)\leq N\). Thus, if
\(M-t>N\), then the quantity inside the maximum
in~\Cref{eq:cut-2} is negative, whereas for \(t=M\) it is
nonnegative. Consequently, the maximum is attained for some
\(t\geq M-N\). Since \(M=2N+1\), we have \(M-N=N+1\).

It follows that, for each \(i\in\{1,2\}\), there exists a longest
common subsequence of \(\widehat{\pi}\) and \(Z\delta_i\) whose symbols
from \(\Sigma\) all belong to
\(U_{N+1}U_{N+2}\cdots U_M\).

\paragraph{Step 2c: constructing the canonical permutation.}
Let \(S_L\) be the set of symbols appearing in
\(U_0U_1\cdots U_N\), and let \(S_R:=\Sigma\setminus S_L\). Define
\[
\tau_L:=U_0U_1\cdots U_N,
\qquad
\tau_R:=U_{N+1}U_{N+2}\cdots U_M,
\]
and set
\[
\pi^\star:=\tau_LZ\tau_R.
\]
We show that \(\pi^\star\) satisfies all four inequalities in the
statement of the lemma.

Fix \(i\in\{1,2\}\), and let \(\rho\) be a longest common subsequence
of \(\widehat{\pi}\) and \(\gamma_iZ\) whose symbols from \(\Sigma\)
all belong to \(U_0U_1\cdots U_N\), as guaranteed by Step~2a. Set
\(\rho_\Sigma:=\rho|_\Sigma\). The string \(\rho_\Sigma\) is a
subsequence of both \(\tau_L\) and \(\gamma_i\), since the construction
of \(\tau_L\) preserves the relative order of all symbols in
\(U_0U_1\cdots U_N\). It follows that
\(\rho_\Sigma Z\) is a common subsequence of
\(\pi^\star=\tau_LZ\tau_R\) and \(\gamma_iZ\). Moreover,
\[
|\rho_\Sigma Z|
=
|\rho_\Sigma|+M
\geq
|\rho|,
\]
because \(\rho\) contains at most \(M\) symbols from \(\Sigma_Z\).
Therefore,
\[
\LCS(\pi^\star,\gamma_iZ)
\geq
\LCS(\widehat{\pi},\gamma_iZ).
\]

Similarly, let \(\rho\) be a longest common subsequence of
\(\widehat{\pi}\) and \(Z\delta_i\) whose symbols from \(\Sigma\) all
belong to \(U_{N+1}U_{N+2}\cdots U_M\), as guaranteed by Step~2b, and
again set \(\rho_\Sigma:=\rho|_\Sigma\). The string
\(\rho_\Sigma\) is a subsequence of both \(\tau_R\) and \(\delta_i\).
Consequently, \(Z\rho_\Sigma\) is a common subsequence of
\(\pi^\star\) and \(Z\delta_i\), and
\[
|Z\rho_\Sigma|
=
M+|\rho_\Sigma|
\geq
|\rho|.
\]
Hence
\[
\LCS(\pi^\star,Z\delta_i)
\geq
\LCS(\widehat{\pi},Z\delta_i).
\]

Finally, Step~1 established that sorting the symbols of \(\Sigma_Z\)
does not decrease any of the relevant LCS lengths. Thus, for each
\(i\in\{1,2\}\),
\[
\LCS(\pi^\star,\gamma_iZ)
\geq
\LCS(\widehat{\pi},\gamma_iZ)
\geq
\LCS(\pi,\gamma_iZ),
\]
\[
\LCS(\pi^\star,Z\delta_i)
\geq
\LCS(\widehat{\pi},Z\delta_i)
\geq
\LCS(\pi,Z\delta_i).
\]
This proves the lemma.
\end{proof}

\section{A General Framework for Boolean CSPs}
\label{sec:boolean-csp-framework}

In this section, we give a general framework for reducing arbitrary Boolean CSPs to Ulam median. While reading~\Cref{sec:3sat-reduction}, a careful reader might notice that the clause-witness gadget is the only part of the construction that
uses the form of a \textsc{3-SAT} clause. The consistency gadget uses
only the variables and labels attached to the symbols, and the
anchor-block lemma does not refer to the 3-SAT formula at all. It is
therefore natural to ask whether we can replace the clause-witness
gadget by one for a different Boolean constraint. Given a general constraint, the key ingredient is a pair of local permutations whose score, when restricted to the symbols
agreeing with an assignment, is large exactly when the
constraint is satisfied by the assignment. As we show below, such a gadget, together
with the same balance assumption, is enough for the reduction to
go through.

Let \(\varphi=(V,\mathcal C)\) be a Boolean CSP instance, where
\(V=\{x_1,\ldots,x_n\}\) and
\(\mathcal C=(C_1,\ldots,C_m)\). For an assignment
\(f:V\to\{0,1\}\), let
\[
\Val_\varphi(f)
:=
\bigl|\{j\in[m]:f\text{ satisfies }C_j\}\bigr|,
\]
and define
\[
\OPT_{\mathrm{CSP}}(\varphi)
:=
\max_{f:V\to\{0,1\}}\Val_\varphi(f).
\]

\paragraph{Local constraint gadgets.}
For each constraint \(C_j\), let \(V_j\subseteq V\) denote its
variables, and let \(\Sigma_j\) be a fresh alphabet of \(B\) symbols,
where \(B\) is a fixed positive integer. The alphabets
\(\Sigma_1,\ldots,\Sigma_m\) are pairwise disjoint, and we write
\[
\Sigma
:=
\Sigma_1\sqcup \cdots\sqcup \Sigma_m,
\qquad
N:=|\Sigma|=Bm.
\]

Each symbol \(\sigma\in\Sigma_j\) is associated with a variable
\(\operatorname{var}(\sigma)\in V_j\) and a label
\(\ell(\sigma)\in\{0,1\}\). As before, a symbol agrees with an
assignment when its associated variable is assigned its label.
For a local assignment \(a:V_j\to\{0,1\}\), define
\[
\Agr_j(a)
:=
\{\sigma\in\Sigma_j:
a(\operatorname{var}(\sigma))=\ell(\sigma)\}.
\]

We require two permutations \(P_j,P_j'\in\mathcal P(\Sigma_j)\)
whose score on this agreement set is large if and only if \(C_j\) is satisfied. More precisely, we want
constants \(A\geq0\) and \(\Delta>0\), independent of \(j\), such
that, for every local assignment \(a\),
\begin{equation}
\label{eq:local-csp-gadget}
\LCS(P_j|_{\Agr_j(a)},P_j'|_{\Agr_j(a)})
=
\begin{cases}
A+\Delta, & \text{if \(a\) satisfies \(C_j\)},\\
A, & \text{otherwise}.
\end{cases}
\end{equation}
We refer to \(B\) as the size of the local gadget and to
\(\Delta\) as its gap. We assume that the labeled alphabets and the
local permutations can be constructed in polynomial time from the
CSP instance.

The clause-witness gadget from \Cref{sec:3sat-reduction} is an example of such a local constraint gadget. It has
\(B=3\), \(A=0\), and \(\Delta=1\).

\paragraph{Global balance.}
For the consistency calculation from the
\textsc{3-SAT} proof to go through, we impose the same balance condition. For a
variable \(x\in V\) and \(b\in\{0,1\}\), define
\[
\Sigma_x^b
:=
\{\sigma\in\Sigma:
\operatorname{var}(\sigma)=x,\ \ell(\sigma)=b\},
\]
and require
\begin{equation}
\label{eq:global-balance}
|\Sigma_x^0|=|\Sigma_x^1|
\qquad
\text{for every }x\in V.
\end{equation}
Write \(r_x:=|\Sigma_x^0|=|\Sigma_x^1|\), so that
\(\sum_{x\in V}r_x=N/2\).

As before, for a global assignment \(f:V\to\{0,1\}\), its agreement set is
\[
\Agr(f)
:=
\{\sigma\in\Sigma:
f(\operatorname{var}(\sigma))=\ell(\sigma)\},
\]
and the balance condition gives
\begin{equation}
\label{eq:global-agreement-size}
|\Agr(f)|
=
|\Sigma\setminus\Agr(f)|
=
\frac N2
=
\frac{Bm}{2}
\qquad
\text{for every }f.
\end{equation}
Thus, the total number of symbols agreeing with an assignment, and
the total number disagreeing with it, are independent of the
assignment.

\paragraph{Constraint-checking permutations.}
To make the local scores add, we place the constraint alphabets in
the same order in both permutations:
\[
\beta_1:=P_1P_2\cdots P_m,
\qquad
\beta_2:=P_1'P_2'\cdots P_m'.
\]

\begin{lemma}
\label{lem:framework-constraint-score}
For every \(S\subseteq\Sigma\),
\[
\LCS(\beta_1|_S,\beta_2|_S)
=
\sum_{j=1}^m
\LCS\bigl(
P_j|_{S\cap\Sigma_j},
P_j'|_{S\cap\Sigma_j}
\bigr).
\]
In particular, for every assignment \(f:V\to\{0,1\}\),
\[
\LCS(\beta_1|_{\Agr(f)},\beta_2|_{\Agr(f)})
=
Am+\Delta\Val_\varphi(f).
\]
\end{lemma}

\begin{proof}
Let \(S_j:=S\cap\Sigma_j\), and let \(\tau\) be a common
subsequence of \(\beta_1|_S\) and \(\beta_2|_S\).
For each \(j\), let \(\tau_j:=\tau|_{\Sigma_j}\). Then
\(\tau_j\) is a common subsequence of \(P_j|_{S_j}\) and
\(P_j'|_{S_j}\). Hence
\[
|\tau|
=
\sum_{j=1}^m|\tau_j|
\leq
\sum_{j=1}^m\LCS(P_j|_{S_j},P_j'|_{S_j}).
\]
Conversely, the blocks occur in the same order in both permutations,
so concatenating a longest common subsequence from each block
attains the right-hand side.

For an assignment \(f\), we have
\(
\Agr(f)\cap\Sigma_j=\Agr_j(f|_{V_j}).
\)
Applying \eqref{eq:local-csp-gadget} in each block and summing gives
the second equality.
\end{proof}

\paragraph{Consistency and anchor gadgets.}
We reuse the variable-consistency permutations from
\Cref{sec:3sat-reduction}. For each \(x\in V\) and \(b\in\{0,1\}\),
let \(B_x^b\) list the symbols of \(\Sigma_x^b\) in a fixed order,
and define
\[
\begin{aligned}
\alpha_1
&:=B_{x_1}^0B_{x_1}^1\cdots B_{x_n}^0B_{x_n}^1,\\
\alpha_2
&:=B_{x_1}^1B_{x_1}^0\cdots B_{x_n}^1B_{x_n}^0.
\end{aligned}
\]
The proof of \Cref{lem:consistency-gadget} applies unchanged, giving
\begin{equation}
\label{eq:framework-consistency-score}
\LCS(\alpha_1|_S,\alpha_2|_S)
=
\sum_{x\in V}
\max\bigl\{
|S\cap\Sigma_x^0|,
|S\cap\Sigma_x^1|
\bigr\}
\end{equation}
for every \(S\subseteq\Sigma\). In particular, the symbols
disagreeing with an assignment form one full label class for each
variable. By global balance,
\begin{equation}
\label{eq:framework-honest-consistency}
\LCS\bigl(
\alpha_1|_{\Sigma\setminus\Agr(f)},
\alpha_2|_{\Sigma\setminus\Agr(f)}
\bigr)
=
\frac N2.
\end{equation}

Let \(\Sigma_Z:=\{z_1,\ldots,z_M\}\) be an alphabet disjoint from
\(\Sigma\), where \(M:=2N+1\), and let \(Z:=z_1\cdots z_M\).
As before, define
\[
L_1:=\beta_1Z,
\qquad
L_2:=\beta_2Z,
\qquad
R_1:=Z\alpha_1,
\qquad
R_2:=Z\alpha_2,
\]
and let the set of input permutations be
\[
\Pi_\varphi:=\{L_1,L_2,R_1,R_2\}.
\]

By \Cref{lem:anchor-block}, we may restrict attention, without loss
of generality, to candidate medians of the form
\(\tau_LZ\tau_R\), inducing a partition
\(\Sigma=S_L\sqcup S_R\). In the \textsc{3-SAT}
construction, \(S_L\) was the proposed true side. Here, it is the
proposed \textit{agreement set}, while \(S_R\) is the proposed \textit{disagreement
set}. A partition induced by an assignment \(f\) has
\[
S_L=\Agr(f),
\qquad
S_R=\Sigma\setminus\Agr(f).
\]
As before, write
\[
q(S):=\LCS(\beta_1|_S,\beta_2|_S),
\qquad
\kappa(S):=\LCS(\alpha_1|_S,\alpha_2|_S)
\]
for the constraint and consistency scores, respectively. The proof
of \Cref{lem:canonical-score} applies unchanged, so for each fixed
partition,
\begin{equation}
\label{eq:framework-partition-score}
\max_{\tau_L,\tau_R}
\sum_{\rho\in\Pi_\varphi}
\LCS(\tau_LZ\tau_R,\rho)
=
4M+N+q(S_L)+\kappa(S_R).
\end{equation}
Since each input permutation has length \(M+N\), this gives
\begin{equation}
\label{eq:framework-partition-optimum}
\OPT_{\mathrm{med}}(\Pi_\varphi)
=
3N-
\max_{\Sigma=S_L\sqcup S_R}
\bigl(q(S_L)+\kappa(S_R)\bigr),
\end{equation}
just as in \eqref{eq:3sat-partition-optimum}.

The local gadget condition determines \(q(\Agr(f))\), but a
proposed agreement set need not arise from an assignment. To compare
its score with that of an assignment, we will delete the selected
symbols that disagree with the assignment. The following lemma
bounds the score lost in doing so. It replaces the clause-witness
counting argument in the proof of
\Cref{thm:3sat-exact-relation}.

\begin{lemma}
\label{lem:framework-constraint-stability}
Let \(P,P'\) be permutations of the same alphabet and let \(S,T\) be
two subsets of that alphabet. Then
\[
\LCS(P|_T,P'|_T)
\leq
\LCS(P|_S,P'|_S)+|T\setminus S|.
\]
\end{lemma}

\begin{proof}
Let \(\tau\) be a longest common subsequence of \(P|_T\) and
\(P'|_T\). Deleting all symbols in \(T\setminus S\) removes at most
\(|T\setminus S|\) symbols and leaves a common subsequence of
\(P|_S\) and \(P'|_S\). Hence
\[
|\tau|-|T\setminus S|
\leq
\LCS(P|_S,P'|_S),
\]
which proves the claim.
\end{proof}

We can now repeat the comparison between an arbitrary partition and
a partition induced by an assignment.

\begin{theorem}
\label{thm:boolean-csp-exact-relation}
The instance \(\Pi_\varphi\) satisfies
\[
\OPT_{\mathrm{med}}(\Pi_\varphi)
=
\left(\frac{5B}{2}-A\right)m
-\Delta\OPT_{\mathrm{CSP}}(\varphi).
\]
\end{theorem}

\begin{proof}
By \eqref{eq:framework-partition-optimum}, it suffices to prove that
\[
\max_{\Sigma=S_L\sqcup S_R}
\bigl(q(S_L)+\kappa(S_R)\bigr)
=
\frac N2+Am+\Delta\OPT_{\mathrm{CSP}}(\varphi).
\]
First, let \(f\) attain \(\OPT_{\mathrm{CSP}}(\varphi)\), and take
\[
S_L=\Agr(f),
\qquad
S_R=\Sigma\setminus\Agr(f).
\]
By \Cref{lem:framework-constraint-score} and
\eqref{eq:framework-honest-consistency},
\[
q(S_L)=Am+\Delta\OPT_{\mathrm{CSP}}(\varphi),
\qquad
\kappa(S_R)=\frac N2.
\]
This proves the lower bound on the maximum combined score.

For the reverse inequality, fix an arbitrary partition
\(\Sigma=S_L\sqcup S_R\). Following the proof of
\Cref{thm:3sat-exact-relation}, we choose an assignment that agrees
with as many symbols of \(S_L\) as possible. Thus, for each
\(x\in V\), choose
\[
f(x)\in\arg\max_{b\in\{0,1\}}|S_L\cap\Sigma_x^b|,
\]
breaking ties arbitrarily, and let
\[
D:=|S_L\setminus\Agr(f)|.
\]
The quantity \(D\) counts the selected symbols that disagree with
the decoded assignment.

Apply \Cref{lem:framework-constraint-stability} to the global
permutations \(\beta_1,\beta_2\), with \(S=\Agr(f)\) and \(T=S_L\).
Together with \Cref{lem:framework-constraint-score}, this gives
\begin{equation}
\label{eq:framework-constraint-soundness}
\begin{aligned}
q(S_L)
&\leq q(\Agr(f))+D\\
&=Am+\Delta\Val_\varphi(f)+D.
\end{aligned}
\end{equation}

It remains to account for these disagreements in the consistency
score. For each variable \(x\), the choice of \(f(x)\) and global
balance give
\[
\begin{aligned}
|S_R\cap\Sigma_x^{1-f(x)}|
&=r_x-|S_L\cap\Sigma_x^{1-f(x)}|\\
&\geq r_x-|S_L\cap\Sigma_x^{f(x)}|\\
&=|S_R\cap\Sigma_x^{f(x)}|.
\end{aligned}
\]
Consequently,
\eqref{eq:framework-consistency-score} and
\eqref{eq:global-agreement-size} give
\begin{equation}
\label{eq:framework-consistency-soundness}
\begin{aligned}
\kappa(S_R)
&=\sum_{x\in V}|S_R\cap\Sigma_x^{1-f(x)}|\\
&=|\Sigma\setminus\Agr(f)|-D\\
&=\frac N2-D.
\end{aligned}
\end{equation}

The possible excess in the constraint score is therefore offset by
the consistency loss. Combining
\Cref{eq:framework-constraint-soundness,eq:framework-consistency-soundness},
we obtain
\begin{align*}
q(S_L)+\kappa(S_R)
&\leq
Am+\Delta\Val_\varphi(f)+D+\frac N2-D\\
&=
Am+\Delta\Val_\varphi(f)+\frac N2\\
&\leq
Am+\Delta\OPT_{\mathrm{CSP}}(\varphi)+\frac N2.
\end{align*}
This proves the upper bound on the maximum combined score.
Substituting into \eqref{eq:framework-partition-optimum} and using
\(N=Bm\) proves the theorem.
\end{proof}

\paragraph{Approximation hardness.}
\Cref{thm:boolean-csp-exact-relation} shows that a gap for the source CSP yields a gap for the resulting Ulam-median instance. More precisely, suppose that, for fixed parameters \(A,B,\Delta\)
and constants \(0\leq s<c\leq1\), it is $\mathrm{NP}$-hard to distinguish such
instances for which
\(
\OPT_{\mathrm{CSP}}(\varphi)\geq cm
\)
from those for which
\(
\OPT_{\mathrm{CSP}}(\varphi)\leq sm.
\)
Then \Cref{thm:boolean-csp-exact-relation} gives the following.

\begin{corollary}
\label{cor:boolean-csp-gap}
Under the preceding assumptions, given a four-permutation Ulam-median instance \(\Pi\), in which each input
permutation has length \(3Bm+1\), it is $\mathrm{NP}$-hard to distinguish between the following two cases.
\begin{itemize}
    \item \emph{(Completeness)} \(
\OPT_{\mathrm{med}}(\Pi)
\leq
\left(\frac{5B}{2}-A-\Delta c\right)m
\)
\item \emph{(Soundness)} \(
\OPT_{\mathrm{med}}(\Pi)
\geq
\left(\frac{5B}{2}-A-\Delta s\right)m.
\)
\end{itemize}
Consequently, it is $\mathrm{NP}$-hard to approximate the Ulam median within any
fixed factor \(\rho\) satisfying
\(
1\leq\rho<
\frac{\frac{5B}{2}-A-\Delta s}
     {\frac{5B}{2}-A-\Delta c}.
\)
\end{corollary}

\begin{proof}
The gap follows directly from
\Cref{thm:boolean-csp-exact-relation}. On a completeness instance, an assignment \(f\)
satisfying at least \(cm\) constraints has
\[
Am+\Delta cm
\leq q(\Agr(f))
\leq|\Agr(f)|
=\frac{Bm}{2}.
\]
Hence
\[
\frac{5B}{2}-A-\Delta c\geq2B>0.
\]
For any \(\rho\) in the stated range, a
\(\rho\)-approximation algorithm would return a permutation of cost
strictly less than
\((5B/2-A-\Delta s)m\) on a completeness instance. On a soundness
instance, every permutation has cost at least this threshold. Such
an algorithm would therefore distinguish the two cases.
\end{proof}

Taking the clause-witness gadget and the balanced
\textsc{Max-3-SAT} gap recovers \Cref{cor:53-52-hardness}. More
generally, a stronger source CSP gap alone does not necessarily give
a stronger Ulam-median hardness factor. The local gadget must also
realize the predicate efficiently: what matters is the gap
\(\Delta\) relative to the fixed overhead
\(\frac{5B}{2}-A\).

\section{An Instantiation for \textsc{Max-E3-Lin-2}}
\label{sec:e3lin-gadget}

We now instantiate the framework of
\Cref{sec:boolean-csp-framework} for the problem
\textsc{Max-E3-Lin-2}. The main motivation for switching from 3-SAT to \textsc{Max-E3-Lin-2} is that the latter allows us to assume a much larger gap. Furthermore, we will show that \textsc{Max-E3-Lin-2} admits efficient local gadgets as well. Both of these factors will contribute to ultimately obtaining a 35/34 factor hardness of approximation for Ulam median, surpassing the 53/52 factor hardness from~\Cref{sec:3sat-reduction}.

Let \(\varphi\) be an instance with \(m\)
equations over \(\mathbb F_2\), each of the form
\[
x\oplus y\oplus z=b,
\qquad b\in\{0,1\},
\]
where \(x,y,z\) are distinct variables. We keep the notation
\(\Val_\varphi\) and \(\OPT_{\mathrm{CSP}}\) from the previous section.
The main contribution of this section is an eight-symbol local gadget that satisfies the conditions required by the framework.

\paragraph{The local gadget.}
Fix an equation
\[
C_j:\qquad x\oplus y\oplus z=b.
\]
We now describe the local permutations for this equation. The permutations will make use of the following eight symbols
\[
\Sigma_j
=
\{
x_0^{(1)},x_0^{(2)},x_1^{(1)},x_1^{(2)},
y_0,y_1,z_0,z_1
\}.
\]
Here, the four \(x\)-symbols
are associated with \(x\), the two \(y\)-symbols with \(y\), and
the two \(z\)-symbols with \(z\). Next, we assign labels to these symbols. For \(p\in\{0,1\}\) and
\(t\in\{1,2\}\), we set
\[
\ell(x_p^{(t)})=p\oplus(1-b),
\qquad
\ell(y_p)=p,
\qquad
\ell(z_p)=p.
\]
In other words, the label of a symbol is always equal to its subscript unless the right-hand side \(b\) of the equation equals zero, in which case the labels of only the \(x\)-symbols are flipped. Define
\[
\begin{aligned}
P_j
&:=
x_0^{(1)}\,y_0\,z_0\,x_1^{(1)}\,x_1^{(2)}\,z_1\,y_1\,x_0^{(2)},\\
P_j'
&:=
x_1^{(1)}\,y_0\,z_1\,x_0^{(1)}\,x_0^{(2)}\,z_0\,y_1\,x_1^{(2)}.
\end{aligned}
\]

\begin{lemma}[Local parity gadget]
\label{lem:e3lin-local-gadget}
For every local assignment \(a:\{x,y,z\}\to\{0,1\}\),
\begin{equation}
\label{eq:e3lin-local-gap}
\LCS(P_j|_{\Agr_j(a)},P_j'|_{\Agr_j(a)})
=
\begin{cases}
3, & \text{if \(a\) satisfies \(C_j\)},\\
2, & \text{otherwise}.
\end{cases}
\end{equation}
Moreover, for each variable in \(C_j\), the gadget contains equally
many symbols of label \(0\) and of label \(1\).
\end{lemma}

\begin{proof}
Fix some local assignment \(a:\{x, y, z\}\to \{0, 1\}\) and write \(a_x:= a(x), a_y:= a(y), a_z:= a(z)\).

Now, among the symbols associated with \(x\), each label occurs twice.
Among those associated with \(y\) or with \(z\), each label occurs
once. This holds for either value of \(b\), proving the balance
assertion. 

Next, we prove the LCS claim. There are two cases, and both involve verification by hand. First, suppose that \(b=1\). Then we have
\[
\Agr_j(a)
=
\{
x_{a_x}^{(1)},x_{a_x}^{(2)},y_{a_y},z_{a_z}
\}.
\]
Number these four symbols \(1,2,3,4\) in their order of appearance
in \(P_j\), and let \(r(a)\) be the permutation of \([4]\) recording
their order in \(P_j'\). A common subsequence corresponds exactly
to an increasing subsequence of \(r(a)\), so
\[
\LCS(P_j|_{\Agr_j(a)},P_j'|_{\Agr_j(a)})
=
\operatorname{LIS}(r(a)).
\]
Depending on the values of \(a_x, a_y, a_z\), we have
\[
\begin{array}{c|c|c}
a_xa_ya_z & r(a) & \operatorname{LIS}(r(a))\\ \hline
000,011,101,110 & 2143 & 2\\
001 & 2314 & 3\\
010 & 1423 & 3\\
100 & 3124 & 3\\
111 & 1342 & 3
\end{array}
\]
Inspecting the table, we can see that when \(b=1\), the LCS is 3 if and only if \(a_x\oplus a_y \oplus a_z =1\), i.e., exactly when the local assignment \(a\) satisfies equation \(C_j\). If it does not satisfy the equation \(C_j\), the LCS is 2.

If \(b=0\), then we have
\[
\Agr_j(a)
=
\{
x_{1-a_x}^{(1)},x_{1-a_x}^{(2)},y_{a_y},z_{a_z}
\}.
\]
Proceeding exactly as above, we obtain
\[
\begin{array}{c|c|c}
a_xa_ya_z & r(a) & \operatorname{LIS}(r(a))\\ \hline
001,010,100,111 & 2143 & 2\\
000 & 3124 & 3\\
011 & 1342 & 3\\
101 & 2314 & 3\\
110 & 1423 & 3
\end{array}
\]

Once again, we see that the LCS is 3 if and only if the local assignment \(a\) satisfies the equation \(C_j\), otherwise it is 2. This completes the proof. 

\end{proof}

\begin{corollary}
\label{cor:e3lin-exact-relation}
Let \(\Pi_\varphi\) be the instance obtained from the above gadgets
by the construction of \Cref{sec:boolean-csp-framework}. Its four
input permutations have length \(24m+1\), and
\[
\OPT_{\mathrm{med}}(\Pi_\varphi)
=
18m-\OPT_{\mathrm{CSP}}(\varphi).
\]
\end{corollary}

\begin{proof}
By \Cref{lem:e3lin-local-gadget}, the local score condition holds
with
\[
B=8,\qquad A=2,\qquad\Delta=1.
\]
Summing the local balance equalities over the constraints containing
each variable gives \eqref{eq:global-balance}. The labeled gadgets
can be constructed in polynomial time, so
\Cref{thm:boolean-csp-exact-relation} gives the claimed identity.
The common permutation length is \(3Bm+1=24m+1\).
\end{proof}

\paragraph{Approximation hardness.}
We now combine H{\aa}stad's gap theorem with the local gadget above
to obtain our main hardness result.

\FourInputHardnessTheorem*

\begin{proof}
Fix \(0<\varepsilon<1/34\). For any fixed
\(0<\eta<1/4\), \Cref{thm:hastad-e3lin-gap} gives a source gap with
\[
c=1-\eta
\qquad\text{and}\qquad
s=\frac12+\eta.
\]
The local gadget has
\[
B=8,\qquad A=2,\qquad \Delta=1.
\]
Therefore, \Cref{cor:boolean-csp-gap} implies hardness within every
factor strictly smaller than
\[
\frac{\frac{5B}{2}-A-\Delta s}
     {\frac{5B}{2}-A-\Delta c}
=
\frac{\frac{35}{2}-\eta}{17+\eta},
\]
and we may choose \(\eta>0\) sufficiently small that this ratio is greater
than \(35/34-\varepsilon\). This proves the theorem.
\end{proof}

We next show that the same gap rules out an additive approximation
scheme.

\MedianAdditiveHardness*

\begin{proof}
Fix \(0<\varepsilon\leq 1/200\), and suppose that there is a
polynomial-time algorithm for Ulam median with additive error at most
\(\varepsilon kd\). Given any instance \(\varphi\) of \textsc{Max-E3-Lin-2} with \(m\geq 4\)
equations, construct the corresponding Ulam-median instance
\(\Pi_\varphi\). It consists of \(k=4\) permutations over an alphabet
of size
\(
d=24m+1.
\)
Since \(m\geq4\), the additive error is at most
\[
\varepsilon kd
=
4\varepsilon(24m+1)
\leq
97\varepsilon m.
\]  
Now, choose \(0<\eta<1/4\) sufficiently small that
\begin{equation}
\label{eq:additive-gap-condition}
2\eta+97\varepsilon<\frac12.
\end{equation}
In the completeness case of
\Cref{thm:hastad-e3lin-gap}, we have
\(\OPT_{\mathrm{CSP}}(\varphi)\geq(1-\eta)m\). Hence
\Cref{cor:e3lin-exact-relation} gives
\[
\OPT_{\mathrm{med}}(\Pi_\varphi)
\leq
(17+\eta)m,
\]
and the algorithm returns a permutation \(\rho\) satisfying
\[
\cost_{\Pi_\varphi}(\rho)
\leq
(17+\eta+97\varepsilon)m.
\]
In the soundness case,
\(\OPT_{\mathrm{CSP}}(\varphi)\leq(1/2+\eta)m\), and therefore every
permutation \(\rho\) satisfies
\[\cost_{\Pi_\varphi}(\rho)
\geq
\OPT_{\mathrm{med}}(\Pi_\varphi)
\geq
\left(\frac{35}{2}-\eta\right)m.
\]
By \eqref{eq:additive-gap-condition},
\(
17+\eta+97\varepsilon
<
\frac{35}{2}-\eta.
\)
Thus, the cost returned by the assumed algorithm distinguishes the
completeness and soundness cases of
\Cref{thm:hastad-e3lin-gap}, a contradiction.
\end{proof}

\section{Hardness of Approximation for Ulam Center}
\label{sec:center-hardness}

We now reduce Ulam median to Ulam center while preserving both the number of input permutations and the optimum value. Together with~\Cref{thm:four-input-hardness}, this yields the claimed hardness of Ulam center.
\CenterHardnessTheorem*

We first describe an approximation factor-preserving reduction from the Ulam median problem to the Ulam center problem.

Throughout this section, we use zero-based indexing: the input
permutations, center permutations, and alphabet copies are indexed by
\(
\mathbb Z_k:=\{0,1,\ldots,k-1\}.
\)

\begin{lemma}[Median-to-center reduction]
\label{lem:median-to-center}
There is a polynomial-time algorithm that, given an instance
\(\Pi=\{\pi_j\}_{j\in\mathbb Z_k}\) of Ulam median over an alphabet
\(\Sigma\) of size \(d\), constructs an instance
\(\mathcal C(\Pi)=\{\sigma_j\}_{j\in\mathbb Z_k}\) of Ulam center
consisting of \(k\) permutations over an alphabet \(\Sigma'\) of size
\(kd\) such that
\(\OPT_{\mathrm{ctr}}(\mathcal C(\Pi))=\OPT_{\mathrm{med}}(\Pi)\).
Moreover, from any center solution for \(\mathcal C(\Pi)\) of radius
at most \(R\), one can compute in polynomial time a permutation
\(\rho\in\mathcal P(\Sigma)\) with \(\cost_\Pi(\rho)\leq R\).
\end{lemma}

The reduction is the same as the median-to-center reduction for the Kendall-tau metric~\cite{biedl2009complexity}, and its applicability to the Ulam metric was also observed in~\cite{CGJ21}. These works focus on polynomial-time reducibility for establishing $\mathrm{NP}$-hardness. For completeness, we provide a detailed proof in~\Cref{sec:reduction-med-center} that explicitly establishes the cost-preservation guarantees stated above. Now, using the above lemma, we show~\Cref{thm:center-hardness}.

\begin{proof}[Proof of~\Cref{thm:center-hardness}]
Suppose that, for some \(0<\varepsilon<1/34\), there is a
polynomial-time \((35/34-\varepsilon)\)-approximation algorithm for
Ulam center with four input permutations. Given any four-input Ulam
median instance \(\Pi\), apply~\Cref{lem:median-to-center} to construct
the four-input center instance \(\mathcal C(\Pi)\). Since the
construction preserves the optimum value, we can compute a solution
\(\rho\) for the Ulam median with
\[
\cost_\Pi(\rho)
\leq
\left(\frac{35}{34}-\varepsilon\right)
\OPT_{\mathrm{ctr}}(\mathcal C(\Pi))
=
\left(\frac{35}{34}-\varepsilon\right)
\OPT_{\mathrm{med}}(\Pi),
\]
contradicting~\Cref{thm:four-input-hardness}.
\end{proof}

Next, we use the same reduction to refute the plausibility of attaining
an additive approximation; more specifically, we show the following.

\CenterAdditiveHardness*

\begin{proof}
Assume for contradiction that there exists a polynomial-time
\(\varepsilon d\)-additive approximation algorithm for the Ulam center
problem with some \(0<\varepsilon\leq1/200\). Given an Ulam median
instance \(\Pi\) consisting of four input permutations, each of length
\(\ell\), we first construct an Ulam center instance
\(\mathcal C(\Pi)\) consisting of \(k=4\) permutations, each of length
\(d=4\ell\). We then run the additive approximation algorithm to
obtain a solution \(\tau\). By \Cref{lem:median-to-center}, from
\(\tau\), we can compute in polynomial time a solution \(\rho\) for
the Ulam median problem of the instance \(\Pi\), satisfying
\begin{align*}
\cost_\Pi(\rho)
&\leq
\rad_{\mathcal C(\Pi)}(\tau)\\
&\leq
\OPT_{\mathrm{ctr}}(\mathcal C(\Pi))+\varepsilon d\\
&=
\OPT_{\mathrm{med}}(\Pi)+4\varepsilon\ell\\
&=
\OPT_{\mathrm{med}}(\Pi)+\varepsilon k\ell,
\end{align*}
contradicting~\Cref{thm:median-additive-hardness}.
\end{proof}

\section{Conclusion}\label{sec:conclusion}

We established explicit constant-factor inapproximability for rank aggregation under the Ulam metric. In particular, unless $\mathrm{P}=\mathrm{NP}$, neither Ulam median nor Ulam center admits a PTAS, even when the instance consists of only four input permutations. In establishing these hardness results, we introduced a general reduction framework from Boolean CSPs to Ulam median. Our current instantiation does not optimize the resulting inapproximability factor, and further refinements of the framework may yield stronger bounds.

A natural next step is to narrow the substantial gap between these lower bounds and the best known approximation algorithms. One promising approach is to extend our framework to non-Boolean CSPs, potentially yielding stronger hardness results. Such an extension requires overcoming additional technical obstacles, which we leave as an intriguing open direction. For Ulam center, determining whether one can beat factor $2$ when the number of inputs is part of the input remains another compelling open direction.

\paragraph*{AI Disclosure. }The authors discovered the initial reduction from \textsc{3-SAT}, establishing $\mathrm{NP}$-hardness for four input permutations. They subsequently used ChatGPT 5.6 Sol to develop an initial hardness-of-approximation result. By inspecting the proof, the authors identified a general reduction framework that is presented in this paper. With further assistance from ChatGPT 5.6 Sol, they then constructed the local gadget used to instantiate this framework to derive the claimed constant factor for the hardness. The authors take full responsibility for the results, proofs, and presentation.

\bibliographystyle{plain}
\bibliography{ref}
\newpage
\appendix

\section{Balanced \textsc{Max-3-SAT} Hardness}
\label{app:balanced-3sat-hardness}

We establish two balanced hardness results used in
\Cref{sec:3sat-reduction}. The first justifies the balance assumption
for the exact decision reduction. The second establishes the balanced
gap needed for hardness of approximation.

\begin{lemma}
\label{lem:balanced-3sat-decision}
There is a polynomial-time transformation that maps any \(3\)-CNF
formula \(\varphi_0\), in which every clause contains exactly three
literal occurrences, to a \(3\)-CNF formula \(\varphi\) such that:
\begin{enumerate}
    \item every clause of \(\varphi\) contains exactly three literal
    occurrences;
    \item every variable of \(\varphi\) has equally many positive and
    negative occurrences;
    \item \(\varphi\) is satisfiable if and only if \(\varphi_0\) is
    satisfiable; and
    \item the number of clauses in \(\varphi\) is linear in the number
    of clauses in \(\varphi_0\).
\end{enumerate}
Consequently, \textsc{3-SAT} remains \(\mathrm{NP}\)-hard under the
first two restrictions.
\end{lemma}

\begin{proof}
For each variable \(x\), let \(p_x\) and \(n_x\) denote the numbers of
positive and negative occurrences of \(x\) in \(\varphi_0\),
respectively. If \(p_x>n_x\), add \(p_x-n_x\) clauses of the form
\[
(\neg x\vee y\vee\neg y),
\]
using a fresh variable \(y\) in each added clause. If \(n_x>p_x\),
add \(n_x-p_x\) clauses of the form
\[
(x\vee y\vee\neg y),
\]
again using a fresh variable \(y\) in each clause.

Clearly, the new formula \(\varphi\) is balanced. Furthermore, each added clause is a tautology. Thus, \(\varphi\) is
satisfiable if and only if \(\varphi_0\) is satisfiable.

If \(\varphi_0\) has \(m_0\) clauses, then
\[
\sum_x |p_x-n_x|
\leq
\sum_x(p_x+n_x)
=
3m_0.
\]
Thus, at most \(3m_0\) clauses are added, and the resulting formula has
at most \(4m_0\) clauses.
\end{proof}
\begin{lemma}
\label{lem:balanced-3sat-gap}
For every fixed \(0<\eta<1/16\), it is $\mathrm{NP}$-hard to distinguish
\(3\)-CNF formulas \(\varphi\) with \(m\) clauses for which
\(
\OPT_{\mathrm{Max3SAT}}(\varphi)\geq(1-\eta)m
\)
from those for which
\(
\OPT_{\mathrm{Max3SAT}}(\varphi)
\leq\left(\frac78+\eta\right)m,
\)
even when every clause contains exactly three distinct variables
and every variable has equally many positive and negative
occurrences.
\end{lemma}

\begin{proof}
Let $m_0$ be the number of linear equations in the \textsc{Max-E3-Lin-2} instance. We apply \Cref{thm:hastad-e3lin-gap} with
\(\delta:=4\eta\). For each equation
\[
x\oplus y\oplus z=b,
\]
there are exactly four assignments to \((x,y,z)\) that violate the
equation. For each such assignment, introduce the clause on
\(x,y,z\) that is false precisely on that assignment. Thus, an
assignment satisfying the equation satisfies all four clauses,
whereas an assignment violating the equation satisfies exactly
three. Moreover, among the four violating assignments, each
variable takes each value exactly twice. Hence each variable occurs
twice positively and twice negatively among the four clauses.

Let \(\varphi\) be the resulting formula, with \(m=4m_0\) clauses.
Every clause contains three distinct variables, and every variable
has equally many positive and negative occurrences. If an assignment
\(f\) satisfies \(s(f)\) equations in the original system, then
\[
\Val_\varphi(f)=3m_0+s(f).
\]
In the completeness case, some assignment therefore satisfies at
least
\[
3m_0+(1-\delta)m_0
=\left(1-\frac\delta4\right)m
=(1-\eta)m
\]
clauses. In the soundness case, every assignment satisfies at most
\[
3m_0+\left(\frac12+\delta\right)m_0
=\left(\frac78+\frac\delta4\right)m
=\left(\frac78+\eta\right)m
\]
clauses. The construction is polynomial-time and proves the lemma.
\end{proof}

\section{Approximation for a Fixed Number of Inputs}
\label{sec:approx-constant}
Let $\Pi=(\pi_1,\ldots,\pi_k)$ be an instance of Ulam median, where each $\pi_i$ is a permutation of an alphabet $\Sigma$ of size $d$. We assume that $k$ is fixed. In this section, we give an
\[
\alpha_k=1+\frac{\lfloor k/2\rfloor-1}{k}
\]
approximation algorithm. In particular, for $k=4$, this gives a $5/4$-approximation. 

For brevity, let $\OPT=\OPT_{\mathrm{med}}(\Pi)$.

\begin{theorem}
\label{thm:main-algo-constant-m}
For every fixed $k\ge2$, there is a deterministic algorithm running in $O(2^{k+1}d^{k+1})$ time, returning a permutation $ \sigma \in \mathcal{P}(\Sigma) $ such that
\begin{equation}\label{eq:main-factor}
  \cost_{\Pi}(\sigma)\le\alpha_k\OPT
\end{equation}
\end{theorem}

For $ k \leq 2 $, returning any permutation from the input gives an optimal solution. Chakraborty, Das, and Krauthgamer~\cite{chakraborty2021approximating} gave an exact algorithm for $k=3$ and a $3/2$-approximation for every fixed $k\geq 4$, based on a dynamic program for an edit-median relaxation. The algorithm in~\Cref{thm:main-algo-constant-m} also uses the same relaxation and dynamic program but refines the rounding step and its analysis. It computes an optimal length-$d$ string over $\Sigma$, which may contain repeated symbols and omit others. It then rounds this relaxed solution to a permutation by removing duplicate occurrences and inserting the missing symbols, while controlling the resulting increase in cost. 

We describe and analyze the algorithm using insertion--deletion edit distance. For strings $u$ and $v$, define
\[
\ED(u,v)=|u|+|v|-2\LCS(u,v).
\]
This definition applies to arbitrary strings, not necessarily permutations. When $u$ and $v$ are permutations of $\Sigma$, both have length $d$, and hence
\[
\ED(u,v)=2(d-\LCS(u,v))=2\U(u,v).
\]

\paragraph{Algorithm.} The algorithm in~\Cref{thm:main-algo-constant-m} consists of the following steps.

\begin{enumerate}[label=Step \arabic*., ref = Step \arabic*,leftmargin=3.5em]
\item \label{enu:median.string}Compute a length-$ d $ string $ \widehat\sigma = \argmin_{x \in \Sigma^{d}} \sum_{i=1}^{k} \ED(x, \pi_{i}) $ using the dynamic programming algorithm in~\cite{chakraborty2021approximating}. For each $ \pi_{i} \in \Pi $, fix an optimal alignment $ \mathcal{A}_{i}$ between $ \widehat \sigma $ and $ \pi_{i} $.

\item \label{enu:remove.repeated.symbols}For each repeated symbol in $ \widehat \sigma $, keep an occurrence used by the largest number of the alignments $ \mathcal{A}_{i} $; delete its other occurrences. Let $ \sigma' $ be the resulting string.
\item \label{enu:add.missing.symbols} Fix one reference input, say $\pi_1$. Obtain a permutation $\sigma$ by inserting all missing symbols into $ \sigma' $ following the order of $ \pi_{1} $, restricted to the symbols in the retained alignment $ \mathcal{A}_{1} $ and the missing symbols.

Return $\sigma$.
\end{enumerate}

\paragraph{Running time.}\ref{enu:median.string} can be implemented in $ O(2^{k+1}d^{k+1}) $ time (see~\cite[Section 5]{chakraborty2021approximating}). ~\ref{enu:remove.repeated.symbols} and~\ref{enu:add.missing.symbols} run in $ O(kd) $ and $ O(d) $ time respectively.  Hence, the total running time is $ O(2^{k+1}d^{k+1}) $.

\paragraph{Approximation Ratio.} Recall that in~\ref{enu:median.string}, we obtain a string $ \widehat\sigma = \argmin_{x \in [\Sigma]^{d}} \sum_{i=1}^{k} \ED(x, \pi_{i}) $. Let \[ E = \sum_{i=1}^{k} \ED(\widehat\sigma, \pi_{i}).\]
Occurrences of repeated symbols are removed in~\ref{enu:remove.repeated.symbols} and missing symbols are inserted in~\ref{enu:add.missing.symbols}. Let
\[
M=d-\lvert \supp (\widehat{\sigma})\rvert.
\]
Thus, $M$ is the number of symbols missing from $\widehat{\sigma}$. Since
$\widehat{\sigma}$ has length $d$, $M$ is also the number of duplicate
occurrences removed in ~\ref{enu:remove.repeated.symbols}.

\begin{lemma}\label{lem:bound.cost.median.string.missing.symbols}
    $ E \leq 2\OPT $ and $ M\leq E/2k $.
\end{lemma}
\begin{proof}
    By definition of $ \widehat\sigma $, $ E \leq \sum_{i=1}^{k} \ED(x, \pi_{i}) $, for any permutation $x\in \mathcal{P}(\Sigma) $. 
    Choose $ x $ to be the permutation attaining $ \OPT_{\mathrm{med}}(\Pi) $, and note that $ \ED(x, \pi_{i}) = 2 \U(x, \pi_{i}) $, it follows that $ E\leq 2\OPT $.

    To bound $ M $, observe that $ M $ is the number of symbols in $ \Sigma $ that do not appear in $ \widehat\sigma $. 
    Therefore, the number of unmatched positions between $ \widehat\sigma $ and each $ \pi_{i} $ must be at least $ M $, namely, $ d - \LCS(\widehat\sigma, \pi_{i}) \geq M $. 
    Thus,
    \begin{align}
        E = \sum_{i=1}^{k} \ED(\widehat\sigma, \pi_{i}) = \sum_{i=1}^{k} 2\left(d-\LCS(\widehat\sigma,\pi_{i}) \right)\geq 2kM, \nonumber
    \end{align}
    and hence, $ M \leq E/2k $.
\end{proof}

\newcommand{\kpar}{\idc{2 \nmid k}}
Let $ \kpar $ be the indicator of the parity of $ k $, namely, $ \kpar = 1 $ if $ k $ is odd, and $ 0 $ otherwise. 
The following lemma shows that~\ref{enu:remove.repeated.symbols} reduces the cost by at least $ \kpar M $.
\begin{lemma}\label{lem:remove.repeated.symbols}
    Let $ \sigma' $ be the string obtained from~\ref{enu:remove.repeated.symbols}, then $ \sum_{i=1}^{k} \ED(\sigma', \pi_{i}) \leq E - \kpar M $.
\end{lemma}
    


\begin{proof}
For each $i\in[k]$, let $\ell_i$ be the number of matched pairs in the fixed optimal alignment $\mathcal A_i$, so that $\ell_i=\LCS(\widehat\sigma,\pi_i)$. Let $r_i$ be the number of these matched pairs whose occurrence in $\widehat\sigma$ is deleted in~\ref{enu:remove.repeated.symbols}, and set $D:=\sum_{i=1}^k r_i$.

Consider any deleted occurrence, and let $t$ be the number of fixed alignments in which it is matched. The retained occurrence of the same symbol is matched in at least $t$ alignments. Since each input permutation contains the symbol exactly once, no alignment can match both occurrences. Hence $2t\leq k$, and summing over the $M$ deleted occurrences gives
\[
D\leq \lfloor k/2\rfloor M.
\]
Restricting $\mathcal A_i$ to the undeleted occurrences gives a common subsequence of $\sigma'$ and $\pi_i$ of length $\ell_i-r_i$. These retained alignments need not be optimal, but they provide the bounds
\[
\LCS(\sigma',\pi_i)\geq \ell_i-r_i.
\]
Since $|\sigma'|=d-M$ and $E=\sum_{i=1}^k(2d-2\ell_i)$, it follows that
\begin{align*}
\sum_{i=1}^k\ED(\sigma',\pi_i)
&\leq \sum_{i=1}^k\bigl(2d-M-2(\ell_i-r_i)\bigr)\\
&=E-kM+2D\\
&\leq E-\bigl(k-2\lfloor k/2\rfloor\bigr)M\\
&=E-\kpar M.
\end{align*}
\end{proof}

\begin{lemma}\label{lem:add.missing.symbols}
    Let $ \sigma $ be the permutation obtained from~\ref{enu:add.missing.symbols}, then $ \sum_{i=1}^{k} \ED(\sigma, \pi_{i}) \leq E + (k-2-\kpar )M $.
\end{lemma}

    

\begin{proof}
Use the notation $\ell_i,r_i,D$ from the proof of~\Cref{lem:remove.repeated.symbols}. The retained part of $\mathcal A_1$ gives a common subsequence of $\sigma'$ and $\pi_1$ of length $\ell_1-r_1$. \ref{enu:add.missing.symbols} inserts all $M$ missing symbols in their $\pi_1$ order into the gaps determined by this subsequence. Thus, the retained matches together with the inserted symbols form a common subsequence of $\sigma$ and $\pi_1$, giving
\[
\LCS(\sigma,\pi_1)\geq \ell_1-r_1+M.
\]
For every other input, insertion preserves all retained matches, so
\[
\LCS(\sigma,\pi_i)\geq \ell_i-r_i
\qquad\text{for }i\in\{2,\ldots,k\}.
\]
Both $\sigma$ and each $\pi_i$ have length $d$. Therefore, using $D\leq\lfloor k/2\rfloor M$, we obtain
\begin{align*}
\sum_{i=1}^k\ED(\sigma,\pi_i)
&=2kd-2\sum_{i=1}^k\LCS(\sigma,\pi_i)\\
&\leq 2kd-2\left(\sum_{i=1}^k(\ell_i-r_i)+M\right)\\
&=E+2D-2M\\
&\leq E+2\bigl(\lfloor k/2\rfloor-1\bigr)M\\
&=E+(k-2-\kpar)M.
\end{align*}
\end{proof}

\begin{proof}[Proof of~\Cref{thm:main-algo-constant-m}]
    It suffices to analyze the approximation factor of the permutation $ \sigma $. 
    Applying~\Cref{lem:add.missing.symbols} and~\Cref{lem:bound.cost.median.string.missing.symbols}, we have
    \begin{align}
        \cost_{\Pi}(\sigma) &= \dfrac{1}{2}\sum_{i=1}^{k} \ED(\sigma, \pi_{i}) \nonumber \\
                            &\leq \dfrac{1}{2}\left(E + (k-2-\kpar )M\right) \nonumber \\
                            &\leq \dfrac{1}{2}E\left(1 + \dfrac{\lfloor k/2 \rfloor - 1}{k}\right) \nonumber \\
                            &\leq \OPT \alpha_{k}. \nonumber
    \end{align}
\end{proof}

\section{Reduction from the Ulam Median to the Ulam Center}
\label{sec:reduction-med-center}
\begin{proof}[Proof of~\Cref{lem:median-to-center}]
For each \(i\in\mathbb Z_k\), let \(\Sigma^{(i)}\) be a disjoint copy
of \(\Sigma\). For a permutation \(\pi\in\mathcal P(\Sigma)\), let
\(\pi^{(i)}\) denote its copy on \(\Sigma^{(i)}\). The center instance
is defined over the alphabet
\(\bigsqcup_{i\in\mathbb Z_k}\Sigma^{(i)}\) as follows. For each
\(j\in\mathbb Z_k\), define
\[
\sigma_j
:=
\pi_j^{(0)}
\pi_{j+_k 1}^{(1)}
\cdots
\pi_{j+_k(k-1)}^{(k-1)},
\]
where the notation \(+_k\) denotes the sum modulo \(k\). Thus, every
\(\sigma_j\) contains one copy of every input permutation, and these
copies are cyclically shifted across the \(k\) blocks. The construction
is illustrated in~\Cref{fig:median-to-center-overview}.

\begin{figure}[t]
\centering
\begin{tikzpicture}[
  font=\small,
  ctrbase/.style={draw,minimum width=1.25cm,minimum height=.50cm,
                  inner sep=1pt,font=\scriptsize},
  ctrone/.style={ctrbase,fill=blue!13},
  ctrtwo/.style={ctrbase,fill=orange!20},
  ctrthree/.style={ctrbase,fill=green!16},
  ctrfour/.style={ctrbase,fill=red!13},
  ctrcandidate/.style={ctrbase,fill=black!7},
  ctrnote/.style={draw,rounded corners,fill=black!3,align=left,
                  text width=5.2cm,inner sep=5pt,font=\small}
]
  \node[font=\scriptsize] at (0,.55) {\(\Sigma^{(0)}\)};
  \node[font=\scriptsize] at (1.35,.55) {\(\Sigma^{(1)}\)};
  \node[font=\scriptsize] at (2.70,.55) {\(\Sigma^{(2)}\)};
  \node[font=\scriptsize] at (4.05,.55) {\(\Sigma^{(3)}\)};

  \node[anchor=east,font=\scriptsize] at (-.75,0) {\(\sigma_0\)};
  \node[ctrone] at (0,0) {\(\pi_0^{(0)}\)};
  \node[ctrtwo] at (1.35,0) {\(\pi_1^{(1)}\)};
  \node[ctrthree] at (2.70,0) {\(\pi_2^{(2)}\)};
  \node[ctrfour] at (4.05,0) {\(\pi_3^{(3)}\)};

  \node[anchor=east,font=\scriptsize] at (-.75,-.62) {\(\sigma_1\)};
  \node[ctrtwo] at (0,-.62) {\(\pi_1^{(0)}\)};
  \node[ctrthree] at (1.35,-.62) {\(\pi_2^{(1)}\)};
  \node[ctrfour] at (2.70,-.62) {\(\pi_3^{(2)}\)};
  \node[ctrone] at (4.05,-.62) {\(\pi_0^{(3)}\)};

  \node[anchor=east,font=\scriptsize] at (-.75,-1.24) {\(\sigma_2\)};
  \node[ctrthree] at (0,-1.24) {\(\pi_2^{(0)}\)};
  \node[ctrfour] at (1.35,-1.24) {\(\pi_3^{(1)}\)};
  \node[ctrone] at (2.70,-1.24) {\(\pi_0^{(2)}\)};
  \node[ctrtwo] at (4.05,-1.24) {\(\pi_1^{(3)}\)};

  \node[anchor=east,font=\scriptsize] at (-.75,-1.86) {\(\sigma_3\)};
  \node[ctrfour] at (0,-1.86) {\(\pi_3^{(0)}\)};
  \node[ctrone] at (1.35,-1.86) {\(\pi_0^{(1)}\)};
  \node[ctrtwo] at (2.70,-1.86) {\(\pi_1^{(2)}\)};
  \node[ctrthree] at (4.05,-1.86) {\(\pi_2^{(3)}\)};

  \draw[dashed,black!55] (-.65,-2.35)--(4.70,-2.35);
  \node[anchor=east,font=\scriptsize] at (-.75,-2.82) {\(\bar\tau\)};
  \node[ctrcandidate] at (0,-2.82) {\(\rho_0^{(0)}\)};
  \node[ctrcandidate] at (1.35,-2.82) {\(\rho_1^{(1)}\)};
  \node[ctrcandidate] at (2.70,-2.82) {\(\rho_2^{(2)}\)};
  \node[ctrcandidate] at (4.05,-2.82) {\(\rho_3^{(3)}\)};
\end{tikzpicture}
\caption{The median-to-center reduction for \(k=4\). Cyclic shifts
place every original input once in every row and column, preserving the
optimum.}
\label{fig:median-to-center-overview}
\end{figure}

We first show that a center candidate can be assumed to respect the
common block order. Let \(\tau\) be an arbitrary permutation of
\(\bigsqcup_{i\in\mathbb Z_k}\Sigma^{(i)}\), and let
\(\overline\tau\) be obtained by stably sorting the symbols of
\(\tau\) according to the block order
\(\Sigma^{(0)},\Sigma^{(1)},\ldots,\Sigma^{(k-1)}\), while preserving
the relative order within each block. Every \(\sigma_j\) has this same
block order. Hence every common subsequence of \(\tau\) and
\(\sigma_j\) already visits the blocks in this order, and stable
sorting preserves such a subsequence. Therefore, for every
\(j\in\mathbb Z_k\),
\[
\U(\overline\tau,\sigma_j)
\leq
\U(\tau,\sigma_j).
\]
For each \(i\in\mathbb Z_k\), let \(\rho_i\in\mathcal P(\Sigma)\)
denote the restriction of \(\overline\tau\) to the \(i\)-th alphabet
copy \(\Sigma^{(i)}\), where the symbols are replaced by their
counterparts in the original alphabet \(\Sigma\). Since
\(\overline\tau\) and \(\sigma_j\) have the same block order, their LCS
decomposes across the blocks, and therefore
\[
\U(\overline\tau,\sigma_j)
=
\sum_{i\in\mathbb Z_k}
\U(\rho_i,\pi_{j+_k i}).
\]
Averaging over \(j\in\mathbb Z_k\) gives
\begin{align*}
\frac{1}{k}
\sum_{j\in\mathbb Z_k}
\U(\overline\tau,\sigma_j)
&=
\frac{1}{k}
\sum_{j\in\mathbb Z_k}
\sum_{i\in\mathbb Z_k}
\U(\rho_i,\pi_{j+_k i})\\
&=
\frac{1}{k}
\sum_{i\in\mathbb Z_k}
\cost_\Pi(\rho_i)\\
&\geq
\OPT_{\mathrm{med}}(\Pi).
\end{align*}
The second equality holds because, for every fixed \(i\), the cyclic
index \(j+_k i\) ranges over \(\mathbb Z_k\) exactly once as \(j\)
ranges over \(\mathbb Z_k\). Since
\[
\max_{j\in\mathbb Z_k}\U(\tau,\sigma_j)
\geq
\max_{j\in\mathbb Z_k}\U(\overline\tau,\sigma_j)
\geq
\frac{1}{k}
\sum_{j\in\mathbb Z_k}
\U(\overline\tau,\sigma_j),
\]
every center candidate \(\tau\) has radius at least
\(\OPT_{\mathrm{med}}(\Pi)\). Hence
\[
\OPT_{\mathrm{ctr}}(\mathcal C(\Pi))
\geq
\OPT_{\mathrm{med}}(\Pi).
\]

For the reverse inequality, let \(\rho^*\) be an optimal median for
\(\Pi\), and consider
\[
\tau^*
:=
(\rho^*)^{(0)}
(\rho^*)^{(1)}
\cdots
(\rho^*)^{(k-1)}.
\]
For every \(j\in\mathbb Z_k\), the block decomposition gives
\[
\U(\tau^*,\sigma_j)
=
\sum_{i\in\mathbb Z_k}
\U(\rho^*,\pi_{j+_k i})
=
\sum_{\ell=0}^{k-1}
\U(\rho^*,\pi_\ell)
=
\OPT_{\mathrm{med}}(\Pi).
\]
Thus, \(\tau^*\) has radius \(\OPT_{\mathrm{med}}(\Pi)\), and therefore
\[
\OPT_{\mathrm{ctr}}(\mathcal C(\Pi))
\leq
\OPT_{\mathrm{med}}(\Pi).
\]
Together with the previous lower bound, this proves equality of the
optimum values.

Finally, suppose that a center solution \(\tau\) has radius at most
\(R\). Construct \(\overline\tau\) and
\(\rho_0,\ldots,\rho_{k-1}\) as above. Since stable sorting does not
increase any distance, \(\U(\overline\tau,\sigma_j)\leq R\) for every
\(j\in\mathbb Z_k\). Therefore,
\[
\frac{1}{k}
\sum_{i\in\mathbb Z_k}
\cost_\Pi(\rho_i)
=
\frac{1}{k}
\sum_{j\in\mathbb Z_k}
\U(\overline\tau,\sigma_j)
\leq R.
\]
Hence at least one \(i\in\mathbb Z_k\) satisfies
\(\cost_\Pi(\rho_i)\leq R\). Computing all \(\rho_i\) and returning
one of minimum cost proves the extraction claim.
\end{proof}

\end{document}